\documentclass[letterpaper, 12pt]{article}
\usepackage[textwidth=17.5cm,textheight=23cm]{geometry}
\usepackage{bbm, amsfonts, amsmath}
\usepackage{amsthm}
\usepackage{graphicx, color}
\usepackage{multirow}
\usepackage{authblk}
\usepackage{setspace}
\usepackage{booktabs}
\usepackage{subcaption}
\usepackage{verbatim}
\usepackage[hidelinks]{hyperref}

\graphicspath{{./figure/}}

\newtheorem{theorem}{Theorem}

\newtheorem{lemma}{Lemma}

\begin{document}

\title{Modeling Excess Mortality and Interest Rates using Mixed Fractional Brownian Motions}

\date{}

\author[1]{Kenneth Q. Zhou}
\author[2]{Hongjuan Zhou\thanks{Corresponding author. E-mail: {\tt Hongjuan.Zhou@asu.edu}}}

\affil[1]{Department of Statistics and Actuarial Science, University of Waterloo, Canada}
\affil[2]{School of Mathematical and Statistical Sciences, Arizona State University, USA}

\maketitle

\textbf{Abstract:}
Recent studies have identified long-range dependence as a key feature in the dynamics of both mortality and interest rates. Building on this insight, we develop a novel bivariate stochastic framework based on mixed fractional Brownian motions to jointly model their long-memory behavior and instantaneous correlation. Analytical solutions are derived under the pricing measure for explicitly valuing zero-coupon bonds and pricing extreme mortality bonds, while capturing the impact of persistent and correlated risk dynamics. We then propose a calibration procedure that sequentially estimates the model and risk premium parameters, including the Hurst parameters and the correlation parameter, using the most recent data on mortality rates, interest rates, and market conditions. Lastly, an extensive numerical analysis is conducted to examine how long-range dependence and mortality-interest correlation influence fair coupon rates, bond payouts and risk measures, providing practical implications for the pricing and risk management of mortality-linked securities in the post-pandemic environment. \\


\section{Introduction}

The accurate modeling of mortality risk plays a fundamental role in the insurance and pension industries, where the valuation and risk management of financial instruments are intricately tied to demographic uncertainties. Mortality-linked securities, such as longevity bonds and catastrophe mortality bonds, serve as vital tools for transferring risk from insurers to the capital markets. These instruments enable stakeholders to manage exposure to extreme demographic events while supporting the stability of financial systems. However, the development of these securities necessitates sophisticated models that account for the stochastic dynamics of mortality rates and their interplay with economic variables, particularly interest rates \cite{biffis2005affine, blackburn2013affine}.

Traditional approaches to mortality modeling often rely on deterministic methods, such as life tables or parametric laws (e.g., \cite{gompertz1825law, makeham1860law}). While these models provide useful baselines for life insurance pricing, they fail to capture the inherent randomness and temporal evolution of mortality rates, especially in the face of systemic demographic shocks. To address these limitations, stochastic mortality models have been developed, such as the Lee-Carter model \cite{lee1992modeling} and the Cairns-Blake-Dowd (CBD) models \cite{cairns2006twofactor, cairns2009comparison}. These seminal works introduced discrete-time stochastic frameworks for mortality forecasting, which have since been extended to continuous-time settings \cite{biffis2005affine, blackburn2013affine}. Such stochastic approaches offer more robust tools for assessing demographic risks over time and are particularly relevant when mortality is coupled with economic factors, such as fluctuating interest rates, which jointly influence the valuation of mortality-linked securities \cite{deelstra2024correlation, milevsky2001mortality}.

Despite significant advancements, existing models frequently overlook certain essential features of mortality and interest rate dynamics. For instance, traditional stochastic models often assume memoryless processes, such as Brownian motion, which lack the persistence observed in real-world data. Empirical evidence suggests that both mortality and interest rates exhibit long-range dependence (LRD), characterized by persistent correlations over extended time horizons \cite{Cajueiro2007, wang2021volterra, yan2019multivariate}. Recent actuarial studies have highlighted the importance of LRD in mortality modeling and demonstrated its critical role in accurately pricing longevity-linked securities and life insurance products \cite{deelstra2024correlation, yan2021longmemory, zhou2023impact}. Frameworks such as the fractional Ornstein-Uhlenbeck process \cite{delgado2019fou} and mixed fractional Poisson processes \cite{jiang2024mfpp} have been proposed to capture these effects, providing a theoretical basis for incorporating long-memory dynamics.

However, many stochastic mortality models fail to jointly model mortality and interest rates while accounting for their correlation and long-memory features \cite{wang2021longevity}. Furthermore, the interdependence between these two dynamics introduces additional complexity, as their joint behavior directly impacts the pricing of mortality-linked financial instruments \cite{deelstra2024correlation}. These valuable findings accentuate the evolving complexities in financial markets and human mortality. In particular, the COVID-19 pandemic has provided two important insights into these dynamics. First, it highlighted how unexpected shocks can drastically alter excess mortality and interest rate trajectories. Second, it demonstrated the long-term impact of health crises on economic conditions and demographic risks, which exposes limitations in existing models, particularly their inability to adapt to systemic shocks. These challenges highlight the need for a more comprehensive stochastic framework that integrates key features that are believed to significantly impact actuarial work, including long-range dependence, short-term volatility, mean-reverting trends, and correlations between mortality and economic factors \cite{jiang2024mfpp, wang2021longevity}.

This paper introduces a novel bivariate stochastic framework to jointly model excess mortality and interest rate dynamics. The mixed fractional Brownian motion that drives our stochastic model can effectively capture the long-range dependence inherent in both mortality and interest rate data. The drift part of our model incorporates mean-reverting trends that are essential for pricing mortality-linked securities. The model also considers short-term volatility and explicitly uses a correlation coefficient matrix to describe the interactions between mortality and interest rates. By addressing these interdependencies and dynamic features, our model framework aims to overcome limitations in existing stochastic mortality models and provide a robust foundation for actuarial and financial applications.

The contributions of this paper are three-fold. The first goal of this paper is to contribute in the research line on the applications of fractional Brownian motion in actuarial work and to explore the valuable insights that fractional models can offer to capture the joint dynamics of interest rate and mortality rate. As mentioned, a novel long-memory bivariate stochastic framework driven by a flexible long-memory process, mixed fractional Brownian motion, is proposed and studied. 
Fractional Brownian motion is known to have more flexibility than Brownian motion that has been traditionally used in actuarial literature. It can effectively capture the long-range dependence that has been proved to exist in financial market and mortality data. Due to the tractability challenges that classic It\^o calculus cannot be directly applied in stochastic models driven by fractional Brownian motion, the research on these models still remains a frontier in probability theory, and it also appears to be an emerging field in applications including both quantitative finance and actuarial science. Through our investigation in this paper, the model we present demonstrates another important application of fractional Brownian motion in actuarial work. It has sufficient flexibility to capture the fore-mentioned features to fill significant gaps in the actuarial literature of stochastic modeling and provides a robust tool for analyzing mortality-interest rate dynamics.

Second, we study the properties of the proposed model, including the probabilities of negative interest rates, the upper and lower bound estimates for the tail probability of extreme mortality rates, and the representation of the model under a risk-neutral pricing measure. Based on these properties, we derive analytical solutions for pricing financial instruments, such as zero-coupon bonds and catastrophe mortality bonds. Specifically, we obtain the value of zero-coupon bonds both at the time of issuance and at any subsequent time point during their lifetime. Furthermore, we investigate the application of the proposed model to price catastrophe mortality bonds, extending its utility to manage extreme mortality risks, similar to the ones considered in \cite{chen2022modeling} and \cite{schmitt2023providing}. These analytical solutions enhance the tractability of the model and provide valuable tools for actuarial and financial applications.

Third, the paper presents a comprehensive calibration methodology using weekly mortality data from the Short-Term Mortality Fluctuations series of the Human Mortality Database (HMD) and interest rate data from the Federal Reserve Economic Data (FRED). 
It should be mentioned that our framework models excess mortality rates, which provides flexibility in the choice of baseline mortality. While \cite{li2023pricing} adopts average mortality rates over the observation period as the baseline, we take a different approach that accounts for seasonality in weekly mortality patterns and allows for adaptability to potential systemic shifts in mortality dynamics.
A sequential calibration procedure is designed to estimate the parameters of the proposed bivariate model, including long-memory Hurst parameters to quantify long-term persistence and the correlation coefficient between mortality and interest rates. Numerical results reveal, similar to findings in \cite{li2023pricing}, that the correlation between mortality and interest rates is negative. Additionally, the proposed model is calibrated under the pricing measure to observed real-world catastrophe mortality bond prices, demonstrating its practical applicability for risk management and insurance pricing.

The paper is organized as follows. In Section 2, the stochastic calculus elements concerning the fractional Brownian motion and the mixed fractional Brownian motion are introduced. In Section 3, we propose our bivariate stochastic models to describe the joint dynamics of interest rates and excess mortality rates. The properties of the proposed model will be studied that include an estimate of the extreme value of excess mortality and the construction of the pricing measure $\mathbb{Q}$. In Section 4, actuarial mathematics is developed with respect to the pricing of mortality-linked securities. In Section 5, we will discuss the methodologies of calibrating our stochastic models under both physical measure $\mathbb{P}$ and pricing measure $\mathbb{Q}$. The simulation-based numerical analysis is applied in Section 6 to investigate the sensitivity of a stylized catastropic mortality bond to the changes in different parameters of our stochastic models. We conclude the paper in Section 7 by summarizing our main results and discussing potential future work.

\section{Stochastic calculus elements}

As the fractional Brownian motion is not Markovian, the classic It\^o calculus cannot be directly applied to study the stochastic integrals with respect to fBm. In this section, we introduce some important elements of stochastic analysis that we will use to analyze the properties of stochastic models driven by fractional Brownian motion.

\subsection{Fractional Brownian motion}
  
We consider a filtered probability space $(\Omega, \mathcal{F}, \mathcal{F}_t, \mathbb{P})$ such that all random variables are well defined and stochastic processes are $\mathcal{F}_t$ adapted. The natural filtration of a stochastic process is taken as the $\mathbb{P}$-completion of the filtration that it generates. For some fixed $T>0$, a (standard) fractional Brownian motion (fBm) on time interval $[0, T]$, $B^H = (B_t^H, t \in [0, T])$ with Hurst parameter $H \in (0, 1)$ is defined as a centered Gaussian process with continuous paths such that $B_0^H = 0$ and $${\rm Cov}(B^H_s, B^H_t) = \mathbb{E}\left(B^H_s B^H_t\right) = \frac{1}{2}\left(s^{2H} + t^{2H} - |s-t|^{2H}\right), \ 0 \leq s, t \leq T.$$
When $H = \frac{1}{2}$, the fBm reduces to the standard Brownian motion (Bm). When $H \neq \frac{1}{2}$, fBm is not a semimartingale and the classic It\^o theory cannot be applied to study the corresponding stochastic integrals. 

The construction of stochastic integrals can follow the traditional limiting arguments (see  \cite{hu19}, \cite{nu} and the references therein). Specifically, we consider the elementary space $\mathcal{E}$ of real valued step functions and define the Hilbert space $\mathfrak{H}$ as the closure of the space $\mathcal{E}$ endowed with the inner product $$\langle \mathbbm{1}_{[0, t]}, \mathbbm{1}_{[0, s]} \rangle_{\mathfrak{H}} = \mathbb{E}(B_s^H B_t^H).$$
The mapping $\mathbbm{1}_{[0, t]} \to B_t^H$ can be extended to  the linear isometry between the Hilbert space and the Gaussian subspace spanned by fBm. In other words, $\int_0^T \mathbbm{1}_{[0, t]} (s) dB_s^H = B^H_t$. By the limiting argument, the stochastic integral (also called the Wiener integral) $ \int_0^T f(t) dB_t^H$ whose integrand $f \in \mathfrak{H}$ is deterministic follows a centered Gaussian distribution. When $H>\frac{1}{2}$, the variance can be computed by
   \begin{equation}\label{var.formula}
       \mathbb{E}\left[\left(\int_0^T f(t) dB_t^H\right)^2\right] = \alpha_H \int_0^T \int_0^T f(s) f(t) |t-s|^{2H-2} ds dt,
   \end{equation}
where $\alpha_H = H(2H-1)$. The covariance between two Wiener integrals is computed by
   \begin{equation}\label{cov.formula}
       {\rm Cov}\left(\int_0^T f(t) dB_t^H, \int_0^T g(t) dB_t^H\right) = \alpha_H\int_0^T \int_0^T f(t) g(s) |t-s|^{2H-2} ds dt.
   \end{equation}

There are some integral transformations that can change fBm to martingales (see \cite{norros1999}). The following results from \cite{norros1999} will be used in this paper. We first introduce the deterministic functions for $0 < s < t \leq T$,
$$k_H(t, s) = \kappa_H^{-1} s^{\frac{1}{2} - H} (t-s)^{\frac{1}{2} - H}, \ \kappa_H = 2H\Gamma(\frac{3}{2}-H)\Gamma(H+\frac{1}{2}),$$
$$\omega_t^H = \lambda_H^{-1} t^{2-2H}, \ \lambda_H = \frac{2H\Gamma(3-2H)\Gamma(H+\frac{1}{2})}{\Gamma(\frac{3}{2}-H)}.$$
Then we can construct a Gaussian martingale $M_t^H = \int_0^t k_H(t, s) dB_s^H$ and it is called the fundamental martingale. The natural filtration generated by this margingale actually coincides with that of the fBm $B^H$, and the quadratic variation $\langle M^H \rangle_t$ equals $\omega_t^H$. The fundamental martingale is very useful in the stochastic analysis about fBm and one of the applications is to derive the Girsanov formula for fBm (see \cite{kleptsyna2000}). Namely, we consider a process $Y_t$ that is the pathwise solution to the following stochastic differential equation $$dY_t = A(t) dt + \sigma(t) dB_t^H, \ 0 \leq t \leq T,$$ where the functions $A(t), \sigma(t)$ satisfy some nice properties to ensure the existence of a pathwise solution for a given initial condition. The process defined by
$$Q_H(t) = \frac{d}{d \omega_t^H} \int_0^t k_H(t, s)\sigma^{-1}(s)A(s) ds, \ t \in [0, T],$$
belongs to the measurable space $L^2([0, T], d\omega^H)$. Then we can construct the random variable $$Z_H(T) = \exp\left(-\int_0^T Q_H(t) dM_t^H - \frac{1}{2} \int_0^T (Q_H(t))^2 d\omega_t^H\right)$$ such that $\mathbb{E}[Z_H(T)] = 1$. Consequently $\mathbb{Q} = Z(T)\mathbb{P}$ is a probability measure, and the process $Y$ under the measure $\mathbb{Q}$ 
satisfies the following equation $$dY_t = \sigma(t) dB_t^{H, \mathbb{Q}}, \ t \in [0, T],$$ where $B_t^{H, \mathbb{Q}}$ is a fBm with Hurst parameter $H$ under measure $\mathbb{Q}$.

\subsection{Mixed fractional Brownian motion}

We will use the combination of fBm and Bm for our proposed stochastic models, as it offers the flexibility of describing different types of driving noise. A mixed fractional Brownian motion (mfBm), $B = B^{\alpha, H} = \{B_t, t \in [0, T]\}$ with parameters $\alpha$ and $H$, was first introduced in \cite{cheridito2001}. It is defined as 
\begin{equation}\label{mfbm.def}
   B_t = \alpha W_t + B_t^H,
\end{equation}
where $\{W_t, t \in [0, T]\}$ is a standard Brownian motion, $\{B_t^H, t \in [0, T]\}$ is a fBm of Hurst parameter $H$, and these two processes are assumed to be independent. The parameter $\alpha \in \mathbb{R}^+$ governs the relative contribution of the standard Brownian component to the mfBm. 

An mfBm inherits many properties from the fBm. It is not a semimartingale if $H \neq \frac{1}{2}$. When $H=\frac{1}{2}$, the mfBm is essentially a Bm. When $H \in (\frac{1}{2}, 1)$, the increment of mfBm exhibits long range dependence (LRD). Namely, we can take a look at the correlation of the increments. Since the mfBm has stationary increments, we can compute $$\rho(k) := \mathbb{E}[B_1(B_{k+1} - B_k)] = \frac{1}{2}[(k+1)^{2H} - 2k^{2H} + (k-1)^{2H}],$$ for any natural number $k \geq 1$. Clearly, when $k$ is large, the correlation $\rho(k)$ behaves like $k^{2H-2}$ and $\sum_k \rho(k) = \infty$ when $H > \frac{1}{2}$. This shows that the mfBm has LRD when $H > \frac{1}{2}$.

To model the joint dynamics of interest rate and mortality rate under the influence of LRD noises, such as the persistent effects of systemic heathcare shocks and medical technology breakthroughs, it is natural to consider the use of mfBm. On the one hand, standard Bm is commonly employed to capture short-term stochastic fluctuations arising from the economic and ecological environment. On the other hand, there is growing empirical evidence that mortality rates and economic data exhibit long-memory behavior, as documented in \cite{zhou2022} and references therein. The LRD reflects the persistent impact of sociodemographic, medical, and environmental factors over extended periods. Similarly, studies such as \cite{li2023pricing} have shown interest rates, under stressed conditions, exhibit sensitivity to health-related shocks. They do not necessarily imply a direct causal relationship but they reveal meaningful empirical association. As noted, the COVID-19 pandemic has reinforced the interdependence and highlighted how systemic events can induce lasting correlated shifts in both financial market and human mortality. These observations suggest a modeling framework based on mfBm, in contrast to traditional Brownian models, offers more realistic and flexible approach for capturing the joint behavior of interest rates and mortality rates. In the next section, we introduce a general version of this framework and present its mathematical formulation.

\section{Model description} \label{sec:ModelDescr}

In this section, we propose our bivariate stochastic models for interest rates and excess mortality rates. The properties of the models and the construction of a risk-neutral measure will be derived.

\subsection{Stochastic models for interest rates and mortality rates}
  
Assume that the bivariate process $\{ X_t, 0 \leq t \leq T \}$ follows a mean-reverting process. Namely, for any $0 \leq t \leq T$,
$$dX_t = (m_t - \Theta_t X_t)dt + \Sigma_{1,t} A_t dW_t + \Sigma_{2,t} dB_t^H,$$
where $\Theta_t, \Sigma_{1,t}, \Sigma_{2,t}, A_t \in \mathbb{R}^{2 \times 2}$, and $m_t \in \mathbb{R}^2$ are all deterministic functions of $t$. $W_t = (W_{1, t}, W_{2, t})'$ for $0 \leq t \leq T$ is a standard bivariate Brownian motion. $B_t^H = (B_{1, t}^{H_1}, B_{2, t}^{H_2})'$ for $0 \leq t \leq T$ is a bivariate fractional Brownian motion with Hurst parameters $H = (H_1, H_2) \in (\frac{1}{2}, 1) \times (\frac{1}{2}, 1)$. We assume that the stochastic process $W_t$ is independent of $B_t$.

By taking $$\Theta_t = \left(\begin{matrix} \theta_1 & 0 \\ 0 & \theta_2 \end{matrix}\right), \ m_t = \left(\begin{matrix} m_1 \\ m_2 \end{matrix} \right), \ A_t = \left(\begin{matrix} \alpha_1 & 0 \\ 0 & \alpha_2 \end{matrix}\right), $$ $$\Sigma_{1,t} = \left(\begin{matrix} \sigma_1 & 0 \\ \frac{\sigma_2 \alpha_2}{\alpha_1} \rho & \sigma_2 \sqrt{1-\rho^2}\end{matrix}\right), \ \Sigma_{2,t} = \left(\begin{matrix} \sigma_1 & 0 \\ 0 & \sigma_2 \end{matrix}\right),$$ we can consider an important special case to model the interest rate process $r_t$ and the excess mortality process $\mu_t$. Equivalently, the bivariate process $(r_t, \mu_t)'$ satisfies the system of stochastic differential equations (SDEs), 
    \begin{eqnarray}
        d r_t & = & (m_1 - \theta_1 r_t) dt + \sigma_1 (\alpha_1 dW_{1, t} + d B_{1, t}^{H_1}), \nonumber\\
	d \mu_t & = & (m_2 - \theta_2 \mu_t) dt + \sigma_2 \left[\alpha_2 \left(\rho dW_{1, t} + \sqrt{1-\rho^2} dW_{2, t} \right) + d B_{2, t}^{H_2}\right]. 
        \label{int.mort.sde}
    \end{eqnarray}
It is essentially a Vasicek-type model. For example, the parameter $\theta_1$ measures the speed of reversion, the quantity $\frac{m_1}{\theta_1}$ measures the long-term mean level of interest rates, and the parameter $\sigma_1$ is the instantaneous volatility. A similar interpretation is applied for the excess mortality rate model.

From the equation system \eqref{int.mort.sde}, we clearly see that the interest rate and excess mortality rate are both driven by a mixed fractional Brownian motion. The dependence between interest rate and excess mortality is captured via the correlation parameter $\rho$. The instantaneous correlation between $r_t$ and $\mu_t$ can be computed using the following lemma.

\begin{lemma}\label{modelP.var.corr}
	Suppose the interest rate $r_t$ and excess mortality rate $\mu_t$ are described by the system of SDEs \eqref{int.mort.sde}. For $H \in (\frac{1}{2}, 1)$, the instantaneous covariance of the increments is given by $${\rm Cov}(d r_t, d \mu_t) = \rho \sigma_1 \sigma_2 \alpha_1 \alpha_2 dt$$
and the instantaneous correlation is given by $${\rm Corr}(d r_t, d \mu_t) \approx \frac{\rho \sigma_1 \sigma_2 \alpha_1 \alpha_2 dt}{\sqrt{\sigma_1^2 \alpha_1^2 dt + \sigma_1^2 (dt)^{2H_1}}\sqrt{\sigma_2^2 \alpha_2^2 dt + \sigma_2^2 (dt)^{2H_2}}}.$$
\end{lemma}

The instantaneous correlation given in Lemma~\ref{modelP.var.corr} demonstrates that Hurst parameters $H_1$ and $H_2$ affect the instantaneous correlation between the interest rate and the excess mortality rate when observed over a small time interval. Specifically, while the instantaneous covariance arising from the correlated Brownian motions is proportional to the time increment $dt$, the instantaneous correlation relates to $dt$ through the variance terms, which contain contributions from both the Brownian and fractional Brownian components. As the time increment $dt$ approaches zero when $H > \frac{1}{2}$, the instantaneous correlation will approach $\rho$. However, for small $dt$, the fractional Brownian motion inflates the variances, which tends to reduce the magnitude of the correlation compared to $\rho$.

In the next two theorems, we discuss the positiveness of interest rates and an estimate of extreme mortality rates. We will use the function $$\nu_t^2(\alpha, H, \theta, \sigma) = \frac{\alpha^2 \sigma^2 \left(1 - e^{-2\theta t}\right)}{2\theta} + \frac{\sigma^2 H(2H-1)}{\theta^{2H} } \left(\gamma(2H-1, \theta t) - e^{-\theta_1 t} \gamma^*(2H-1, \theta t)\right),$$
where the incomplete gamma functions are given by $$\gamma(u, x) = \int_0^x s^{u-1} e^{-s} ds, \quad \gamma^*(u, x) = \int_0^x (x-s)^{u-1} e^{-s} ds.$$

\begin{theorem}
   Suppose that the interest rate $r_t$ is described by the system of SDEs \eqref{int.mort.sde}. Then the random variable $r_t$ follows the normal distribution with mean $\frac{m_1}{\theta_1} + e^{-\theta_1 t}\left(r_0 - \frac{m_1}{\theta_1}\right)$ and variance $\nu_t^2(\alpha_1, H_1, \theta_1, \sigma_1)$.
   In addition, the probability that $r_t$ is nonnegative is given by 
   \begin{equation}\label{rt.pos}
       P(r_t \geq 0) = \Phi\left(\frac{\frac{m_1}{\theta_1} + e^{-\theta_1 t}\left(r_0 - \frac{m_1}{\theta_1}\right)}{\nu_t(\alpha_1, H_1, \theta_1, \sigma_1)}\right),
   \end{equation}
   where $\Phi(\cdot)$ denotes the cumulative distribution function of the standard normal distribution.
\end{theorem}

\begin{proof}
   The pathwise solution to the SDE for $r_t$ with initial condition $r_0$ is
     $$r_t = \frac{m_1}{\theta_1} + e^{-\theta_1 t}\left(r_0 - \frac{m_1}{\theta_1}\right) + \sigma_1 \int_0^t e^{-\theta_1(t-s)} d(\alpha_1 W_{1,s} + B_{1,s}^{H_1}). $$
     Clearly, it is a Gaussian random variable with mean $\frac{m_1}{\theta_1} + e^{-\theta_1 t}\left(r_0 - \frac{m_1}{\theta_1}\right)$ according to the property of stochastic integrals. To compute the variance, we use the formula \eqref{var.formula} to derive
     \begin{eqnarray*}
         {\rm Var}[r_t] 
          & = & {\rm Var}\left[\sigma_1 e^{-\theta_1 t} \left(\int_0^t \alpha_1 e^{\theta_1 s} dW_{1,s} + \int_0^t e^{\theta_1 s} dB_{1, s}^{H_1}\right)\right] \\
          & = & \sigma_1^2 e^{-2\theta_1 t} \left(\alpha_1^2 \int_0^t e^{2\theta_1 s} ds + H_1(2H_1-1)\int_0^t \int_0^t e^{\theta_1 (u+v)} |u-v|^{2H_1-2} du dv\right) \\
          & = & \frac{\alpha_1^2 \sigma_1^2 \left(1 - e^{-2\theta_1 t}\right)}{2\theta_1} + \frac{\sigma_1^2 H_1(2H_1-1)}{\theta_1^{2H_1}} \left(\int_0^{\theta_1 t} e^{-x} x^{2H_1-2} dx - \int_0^{\theta_1 t} e^{-2\theta_1 t + x} x^{2H_1-2} dx\right),
     \end{eqnarray*}
     where we have made change of variables $x=\theta_1(u-v)$, $y=\theta_1 v$ when evaluating the double integral. As $r_t$ is Gaussian, it is clear that 
     $$P(r_t \geq 0) = \Phi\left(\frac{\frac{m_1}{\theta_1} + e^{-\theta_1 t}\left(r_0 - \frac{m_1}{\theta_1}\right)}{\sqrt{\frac{\alpha_1^2 \sigma_1^2 \left(1 - e^{-2\theta_1 t}\right)}{2\theta_1}} + \frac{\sigma_1^2 H_1(2H_1-1)}{\theta_1^{2H_1}} \left(\gamma(2H_1-1, \theta_1 t) -  e^{-\theta_1 t} \gamma^*(2H_1-1, \theta_1 t)\right)} \right),$$
     which is consistent with \eqref{rt.pos}.
\end{proof}
The variance function $\nu_t^2$ is an increasing function of time $t$. It should be noted that as $t \to \infty$, we have $$\lim_{t \to \infty} \nu_t^2 = \frac{\alpha_1^2 \sigma_1^2}{2\theta_1} + \frac{\sigma_1^2 \Gamma(2H_1+1)}{2\theta_1^{2H_1}},$$ and moreover, when $t$ is large, 
  $$P(r_t \geq 0) \approx \Phi\left(\frac{\frac{m_1}{\theta_1}}{\sqrt{\frac{\alpha_1^2 \sigma_1^2}{2\theta_1} + \frac{\sigma_1^2 \Gamma(2H_1+1)}{2\theta_1^{2H_1}}}}\right) = \Phi\left(\sqrt\frac{2m_1^2}{\alpha_1^2 \sigma_1^2 \theta_1 + \sigma_1^2 \Gamma(2H_1+1)\theta_1^{2-2H_1}}\right).$$
During the modeling process, if a negative interest rate is believed to be rare, we have to carefully choose the applicable time interval for the interest rate model so that the probability of interest rate being positive is sufficiently close to one.

Similarly, we can see that the excess mortality rate $\mu_t$ given by \eqref{int.mort.sde} is also a Gaussian process with mean $\mathbb{E}[\mu_t] = \frac{m_2}{\theta_2} + e^{-\theta_2 t}\left(\mu_0 - \frac{m_2}{\theta_2}\right)$ and variance ${\rm Var}[\mu_t] = \nu_t^2(\alpha_2, H_2, \theta_2, \sigma_2)$. Using properties of Gaussian processes, we can find an upper and a lower bound for the tail probability of extreme excess mortality rates in the time interval $[0, T]$. This is stated in the following theorem.
\begin{theorem}
   The maximum value of the excess mortality rate $\mu_t$ defined by \eqref{int.mort.sde} in the time interval $[0, T]$ satisfies the following inequality
   \begin{equation}\label{max.lower}
       P\left(\sup_{0 \leq t \leq T} \mu_t \geq a + \mathbb{E}[\mu_t] \right) \geq \int_{c*}^\infty \sqrt{\frac{2}{\pi}} e^{-\frac{x^2}{2}} dx,
   \end{equation}
   where $c* = a \nu_T^{-1}(\alpha_2, H_2, \theta_2, \sigma_2).$ Moreover,
   \begin{equation}\label{max.upper}
      P\left(\sup_{0 \leq t \leq T} \mu_t \geq a + \mathbb{E}[\sup_{0 \leq t \leq T} \mu_t] \right) \leq \exp\left(-\frac{a^2}{2\nu_T^2(\alpha_2, H_2, \theta_2, \sigma_2)}\right).
   \end{equation}
\end{theorem}

\begin{proof}
    We follow the proof idea in \cite{ruzmaikina} by using Slepian’s lemma (see \cite{slepian}) to prove this theorem.
    Denote $$X(t) = \alpha_2 \int_0^t e^{-\theta_2(t-s)} d( \rho dW_{1, s} + \sqrt{1-\rho^2} dW_{2, s}) + \int_0^t e^{-\theta_2(t-s)} d B_{2,s}^{H_2}.$$ By the property of mixed fractional Brownian motion, we know $\mathbb{E}[X_t]=0$ and apply the formula \eqref{cov.formula} to obtain the covariance
    $$\gamma_{t, s} := \mathbb{E}[X_t X_s]= \int_0^t \int_0^s e^{-\theta_2(t-u)} e^{-\theta_2(s-v)} |u-v|^{2H_2-2} du dv + \alpha_2^2 \int_0^{t\wedge s} e^{-\theta_2(t+s-2u)} du.$$
    Next we introduce another continuous Gaussian process $\widehat X_t$ with $\mathbb{E}[\widehat X_t] = 0$ and $\mathbb{E}[\widehat X_t \widehat X_s] = \gamma_{s, s}$ for $s \leq t$. By this construction, $\widehat X_t$ is a Markov process. By the reflection principle, we have the following estimate
      $$P(\max_{0 \leq t \leq T} \widehat X_t \geq k) = 2P(Y(T) \geq k) = \int_{k/\sqrt{\gamma_{T,T}}}^\infty \sqrt{\frac{2}{\pi}} e^{-\frac{x^2}{2}} dx.$$
    Clearly, $\mathbb{E}[X_t^2] = \mathbb{E}[\widehat X_t^2]$ and we can use the results of Lemma 18 in \cite{hu19} and Cauchy-Schwarz inequality to verify that  and $\mathbb{E}[X_t X_s] \leq \mathbb{E}[\widehat X_t \widehat X_s]$. Thus, the processes $X(t)$ and $\widehat X(t)$ satisfy the assumptions of Slepian’s lemma. Applying Slepian’s lemma yields
   $$P\left(\max_{0 \leq t \leq T} \mu_t  \geq a + \mathbb{E}[\mu_t]\right) = P\left(\max_{0 \leq t \leq T} \sigma_2 X(t) \geq a\right) \geq P(\max_{0 \leq t \leq T} \widehat X_t \geq a \sigma_2^{-1}) = \int_{\frac{a} {\sigma_2\sqrt{\gamma_{T,T}}}}^\infty \sqrt{\frac{2}{\pi}} e^{-\frac{x^2}{2}} dx.$$
   Taking into account of $\nu_T^2(\alpha_2, H_2, \theta_2, \sigma_2) = \sigma_2^2 \gamma_{T, T}$, we obtain the inequality \eqref{max.lower}.
   For the second inequality \eqref{max.upper}, we can apply Borell-TIS inequality (see \cite{adler2007}, \cite{Borell1975}, \cite{MarcusShepp1972}). Namely, the fact that $\mu_t$ is a bounded Gaussian process in the time interval $[0, T]$ yields
     $$P\left(\sup_{t \in [0, T]} X_t - \mathbb{E}\left[\sup_{t \in [0, T]} X_t\right] \geq a\right) \leq \exp\left(-\frac{a^2}{2 \sup_{t \in [0, T]} \gamma_{t, t}}\right) = \exp\left(-\frac{a^2}{2 \gamma_{T, T}}\right).$$
   This completes the proof of the theorem.
\end{proof}

Next, we study the change of measure for the special model \eqref{int.mort.sde}. Motivated by \cite{li2023pricing}, we take the simplest choice for the market price of diffusion risks, that is, 
  \begin{eqnarray}
	  dW_{i, t} & = & dW_{i, t}^{\mathbb{Q}} + \gamma_i dt, \ i = 1, 2, \nonumber \\
	  dB_{j, t}^{H_j} & = & dB_{j, t}^{H_j, \mathbb{Q}} + \eta_j dt, \ j=1, 2, \label{meas.q}
  \end{eqnarray} 
for $t \in [0, T]$. The parameters $\gamma_i$ and $\eta_j$ specify the market prices of the risks driven by the Brownian motions and the fractional Brownian motions, respectively. The existence and construction of the probability measure $\mathbb{Q}$ are specified in the following theorem and its proof. Note that this construction of the pricing measure is restricted to a finite-time horizon $[0, T]$, but it suffices for our actuarial study.

\begin{theorem}\label{prob.measure.q}
  There exists a probability measure $\mathbb{Q}$ such that the system of equations \eqref{meas.q} is valid on $[0, T]$. Under the pricing measure $\mathbb{Q}$, $(W_{1, t}^{\mathbb{Q}}, W_{2, t}^{\mathbb{Q}})$ is a standard bivariate Brownian motion, and $(B_{1, t}^{H_1, \mathbb{Q}}, B_{2, t}^{H_2, \mathbb{Q}})$ is a two-dimensional fBm with Hurst parameters $(H_1, H_2)$ respectively. These two multivariate processes are mutually independent. Then the model in \eqref{int.mort.sde} on time interval $[0, T]$ can be rewritten as
  \begin{eqnarray}
	  d r_t  & = & (m'_1 - \theta_1 r_t) dt + \sigma_1 \left(\alpha_1 dW_{1, t}^{\mathbb{Q}} + d B_{1, t}^{H_1, \mathbb{Q}}\right), \nonumber \\
	 d \mu_t & = & (m'_2 - \theta_2 \mu_t) dt + \sigma_2 \left[\alpha_2 \left(\rho dW_{1, t}^{\mathbb{Q}} + \sqrt{1-\rho^2} dW_{2, t}^{\mathbb{Q}} \right) + d B_{2, t}^{H_2, \mathbb{Q}}\right], \label{int.mort.qsde}
  \end{eqnarray}
where $m'_1 = m_1 + \alpha_1 \sigma_1 \gamma_1 + \sigma_1 \eta_1$ and $m'_2 = m_2 + \alpha_2 \sigma_2 \rho \gamma_1 + \alpha_2 \sigma_2 \sqrt{1 - \rho^2} \gamma_2 + \sigma_2 \eta_2$.
\end{theorem}

\begin{proof}
Take  $$Q_{H_i}(t) = \frac{d}{d \omega_t^{H_i}} \int_0^t k_{H_i} (t, s) \eta_i ds, \ t \in [0, T],$$
where $$k_{H_i}(t, s) = \kappa_{H_i}^{-1} s^{\frac{1}{2} - H_i} (t-s)^{\frac{1}{2} - H_i}, \ \omega_t^{H_i} = \lambda_{H_i}^{-1} t^{2-2H_i}, \ i = 1, 2.$$
Define the fundamental Gaussian martingales $M_t^{H_i} = \int_0^t k_{H_i}(t, s) dB_s^{H_i}$ under probability measure $\mathbb{P}$. Clearly, the two martingales $M_t^{H_i}$ are mutually independent, and $\langle M^{H_i} \rangle_t = \omega_t^{H_i}$ (see \cite{kleptsyna2000} and the references therein). Then we can define the random variables 
$$Z_1(T) = \exp\left(\gamma_1 W_{1, T} + \gamma_2 W_{2, T} - \frac{1}{2}(\gamma_1^2 + \gamma_2^2) T\right),$$
$$Z_2(T) = \exp\left(\sum_{i=1}^2 \int_0^T Q_{H_i}(t) dM_t^{H_i} - \frac{1}{2} \sum_{i=1}^2 \int_0^T Q_{H_i}(t)^2 d\omega_t^{H_i}\right).$$
They are independent, and we have $\mathbb{E}^{\mathbb{P}}[Z_1(T)] = 1$ and $\mathbb{E}^{\mathbb{P}}[Z_2(T)] = 1$. The pricing measure $\mathbb{Q}$ can be defined by $\mathbb{Q} = Z_1(T)Z_2(T) \mathbb{P}$ where we take the Radon-Nikodym derivative as $Z_1(T)Z_2(T)$. 

Under this construction, we claim that $W^{\mathbb{Q}}(t) = (W^{\mathbb{Q}}_{1,t}, W^{\mathbb{Q}}_{2,t})$ defined in \eqref{meas.q} is a bivariate Brownian motion and $B_{j,t}^{H_j, \mathbb{Q}}$, for $j=1, 2$ defined in \eqref{meas.q} are independent fractional Brownian motions under the probability measure $\mathbb{Q}$. In fact, it suffices to show that the finite-dimensional distribution of $B^{H_j}$ with respect to the measure $\mathbb{Q}$ is fractional Brownian motion. We take the time points $t_1, \ldots, t_l$ for $B^{H_1}$ and $s_1, \ldots, s_k$ for $B^{H_2}$. Let $\Sigma_1 = {\rm Cov}(B^{H_1}_{1,t_i}, B^{H_1}_{1,t_j})_{i, j = 1, \ldots, l}$ and $\Sigma_2 = {\rm Cov}(B^{H_2}_{2,s_i}, B^{H_2}_{2,s_j})_{i, j = 1, \ldots, k}$. Then for any constants $\alpha_i$'s and $\beta_i$'s,
\begin{eqnarray}
    & & \mathbb{E}^{\mathbb{Q}}\exp\left(\sum_{i=1}^l \alpha_i B^{H_1, \mathbb{Q}}_{1,t_i} + \sum_{i=1}^k \beta_i B^{H_2, \mathbb{Q}}_{2,s_i}\right)  \notag\\
    & = & \mathbb{E}^{\mathbb{P}}\left[Z_1(T) \exp\left(\sum_{i=1}^l \alpha_i (B^{H_1}_{1,t_i} - \eta_1 t_i) + \sum_{i=1}^k \beta_i (B^{H_2}_{2,s_i} - \eta_2 s_i) +\sum_{i=1}^2 \int_0^T Q_{H_i}(t) dM_t^{H_i} \right.\right. \notag \\
    & & \left.\left. - \frac{1}{2} \sum_{i=1}^2 \int_0^T Q_{H_i}(t)^2 d\omega_t^{H_i}\right)\right] \notag \\
    & = & \mathbb{E}^{\mathbb{P}}\left[\exp\left(Z - \frac{1}{2} \sum_{i=1}^2 \int_0^T Q_{H_i}(t)^2 d\omega_t^{H_i} - \sum_{i=1}^l \alpha_i t_i \eta_1 - \sum_{i=1}^k \beta_i s_i \eta_2\right)\right],
    \label{rn.proof.eq1}
\end{eqnarray}
where $Z = \sum_{i=1}^l \alpha_i B^{H_1}_{1,t_i} + \sum_{i=1}^k \beta_i B^{H_2}_{2,s_i}+\sum_{i=1}^2 \int_0^T Q_{H_i}(t) dM_t^{H_i}$ is Gaussian with mean zero. The variance of $Z$ can be calculated in the following way:
\begin{equation}\label{rn.proof.eq2}
    \mathbb{E}[Z^2] = {\alpha} \Sigma_1 {\alpha'} + {\beta} \Sigma_2 {\beta'} + \sum_{i=1}^2 \int_0^T Q_{H_i}(t)^2 d\omega_t^{H_i} + 2 \sum_{i=1}^l \alpha_i \eta_1 t_i + 2 \sum_{i=1}^k \beta_i \eta_2 s_i,
\end{equation}
where we have used the facts (see Lemma 1 in \cite{kleptsyna2000})
  $${\rm Cov}\left(B^{H_1}_{1,t_i}, \int_0^T Q_{H_i}(t) dM_t^{H_i}\right) = t_i \eta_1,  \quad {\rm Cov}\left(B^{H_2}_{1,t_i}, \int_0^T Q_{H_i}(t) dM_t^{H_i}\right) = t_i \eta_2.$$
Now we use equations \eqref{rn.proof.eq1} and \eqref{rn.proof.eq2} to obtain
\begin{equation}
    \mathbb{E}^{\mathbb{Q}}\exp\left(\sum_{i=1}^l \alpha_i B^{H_1, \mathbb{Q}}_{1,t_i} + \sum_{i=1}^k \beta_i B^{H_2, \mathbb{Q}}_{2,s_i}\right) = \exp\left(\frac{1}{2}{\alpha} \Sigma_1 {\alpha'} + \frac{1}{2} {\beta} \Sigma_2 {\beta'}\right).
\end{equation}
This shows that $B^{H_1, \mathbb{Q}}$ and $B^{H_2, \mathbb{Q}}$ are independent fractional Brownian motions. For the Brownian motion case, one can follow the similar arguments as above, or choose to follow the proof lines in \cite{li2023pricing} and the reference therein. Finally, it would be straightforward to use equations in \eqref{meas.q} to show that the system of equations \eqref{int.mort.qsde} is valid.
\end{proof}

\section{Actuarial Valuation Framework} \label{sec:PricFram}

The analytic results on our proposed stochastic models discussed in the previous section can be applied to address the pricing challenges of many novel financial instruments that can be used in actuarial risk management. Namely, in this section we present the valuation results for zero-coupon bonds and mortality-linked securities. We explain how interest rate dynamics affect cash flow discounting and describe the design and actuarial valuation of mortality bonds that are useful to manage excess mortality risks.

\subsection{Valuation of zero-coupon bond}

Zero-coupon bonds play a crucial role in the design of mortality-linked securities by serving as instruments for discounting future cash flows. To model this, we assume that the dynamics of the interest rate follows the process given in \eqref{int.mort.qsde}. At time $t$, the value of a risk-free zero-coupon bond maturing at time $T$ is given by \footnote{We acknowledge that there are potential arbitrage opportunities in the fractional bond market as discussed in the current literature (see, for example, \cite{Cheridito2003}). These issues could be addressed by implementing appropriate restrictions on trading strategies. For further discussion, we refer the interested reader to \cite{czichowsky2017}, \cite{Guasoni2010}, \cite{jarrow09}, \cite{mishura2020} and the references therein for further details. Nevertheless, the market for mortality-linked securities, which is the focus of this paper, is generally regarded as incomplete.}
$$P(t, T) = \mathbb{E}^{\mathbb{Q}}\left[\exp\left(-\int_t^T r_s ds\right)\Big|\mathcal{F}_t\right].$$
To calculate the conditional expectation, we need an important result that provides an explicit formula for conditional expectation for a functional of fBm. This result is recalled in the following lemma \ref{lem.condexp}. For any appropriate function $f$, we first introduce an operator $(-\Delta)^\alpha$ defined as
$$(-\Delta)^\alpha f(x) := \frac{1}{2\Gamma(-2\alpha)\cos(-\alpha \pi)} \int_{-\infty}^\infty |x-t|^{-2\alpha-1} f(t) dt,$$
for $\alpha<0$, and the norm of $f$ on the Hilbert space $\mathfrak{H}$ is denoted by
      $$\|f\|_{\mathfrak{H}}^2 := \alpha_H \int_0^\infty \int_0^\infty f(u)f(v)|u-v|^{2H-2} dudv,$$
	where $\alpha_H = H(2H-1)$ and $H \in (\frac{1}{2}, 1)$. 
According to the results given in \cite{gn96} and Theorem 3.2 in \cite{hu01}, the conditional expectation of a functional of fBm with respect to the filtration generated by the fBm, can be represented as a stochastic integral with respect to the same fBm. 
\begin{lemma}\label{lem.condexp}
   Suppose a given function $g$ satisfies necessary conditions such that the stochastic integral $\int_0^\infty g(s) dB^H(s)$ is well defined, where $B^H$ is a fBm with Hurst parameter $H$. Then there exists a function $h$ with support on $[0, t]$, such that 
    \begin{equation*}
		\mathbb{E}\left[\int_0^\infty g(s) dB^H(s) \Big| \mathcal{F}_t\right] = \int_0^t h(s) dB^H(s),
    \end{equation*}
and the function $h$ satisfies
    \begin{equation}\label{h.def}
	(-\Delta)^{\frac{1}{2}-H} g(x) = (-\Delta)^{\frac{1}{2}-H} h(x),
    \end{equation}
for any $x \in (0, t)$. Moreover,
    $$\mathbb{E}\left[\exp\left(\int_0^\infty g(s) dB^H(s) \right)\Big| \mathcal{F}_t\right] = \exp\left(\int_0^t h(s) dB^H(s) - \frac{1}{2}\|h_{[0, t]}\|_{\mathfrak{H}}^2 + \frac{1}{2}\|g\|^2_{\mathfrak{H}}\right).$$
\end{lemma}
	
Applying the above result, we can derive an explicit formula for the bond values. This is stated in the following theorem.

\begin{theorem}\label{risk.free.bond}
   Consider a risk-free zero-coupon bond issued at time $0$ and matured at time $T$, whose interest rate $r_t$ under measure $\mathbb{Q}$ is described by the following SDE driven by a mfBm $B_t = \alpha W_t + B_t^H$:
   $$dr_t = (m - \theta r_t) dt + \sigma dB_t, \ t \in [0, T].$$ The bond value at time 0 can be determined by
    \begin{eqnarray}
	P(0, T) & = & \exp\left(-\frac{m}{\theta}T + \left(\frac{m}{\theta} - r_0\right)\frac{1-e^{-\theta T}}{\theta}\right) \times \exp \left(\frac{\alpha^2 \sigma^2}{2\theta^2} \left(-\frac{1}{2\theta} e^{-2\theta T} + \frac{2}{\theta}e^{-\theta T} + T - \frac{3}{2\theta}\right)\right) \notag\\
	& & \times \exp \Big(\frac{\sigma^2 \alpha_H}{\theta^2} \Big(-\frac{e^{-2\theta T}}{2\theta} \int_0^T e^{\theta x} x^{2H-2} dx + \frac{1}{2\theta}\int_0^T e^{-\theta x} x^{2H-2} dx \notag\\
	& & \ - \frac{1}{2H-1} \int_0^T e^{-\theta x} x^{2H-1} dx - \frac{e^{-\theta T}}{2H-1} \int_0^T e^{\theta x}x^{2H-1} dx + \frac{T^{2H}}{(2H)(2H-1)} \Big) \Big). \label{bond.p0T}
    \end{eqnarray}
    The bond value at time $t$ can be computed as
	  \begin{eqnarray}
		  P(t, T) & = & \exp\left(\int_0^t r_s ds -\frac{mT}{\theta} + \frac{1}{\theta}\left(\frac{m}{\theta} - r_0\right)\left(1 - e^{-\theta T}\right)\right) \times \exp\left(-\frac{\alpha \sigma}{\theta} \int_0^t 1 - e^{-\theta(T-u)} dW_u\right) \notag\\
		  &&  \times \exp\left(\frac{\alpha^2 \sigma^2}{2\theta^2} \left((T-t) - \frac{e^{-2\theta(T-t)} - 4 e^{-\theta(T-t)} + 3}{2\theta}\right)\right) \notag\\
		  && \times \exp\left(\int_0^t h(s) dB_s^H - \frac{1}{2}\|h_{[0, t]}\|_{\mathfrak{H}}^2 + \frac{1}{2}\|g\|^2_{\mathfrak{H}}\right), \label{bond.ptT}
	  \end{eqnarray}
    where the function $h$ satisfies the equation \eqref{h.def} associated with the function $g$ that is defined by
	$$g(u) = -\frac{\sigma}{\theta}\mathbbm{1}_{[0, T]} \left(1 - e^{-\theta(T-u)}\right).$$
\end{theorem}
	
\begin{proof}
For a given initial condition $r_0$, the interest rate process $r_t$ admits the following pathwise solution 
    \begin{equation}\label{rt.sol} 
	r_t = \frac{m}{\theta} + e^{-\theta t}\left(r_0 - \frac{m}{\theta}\right) + \sigma \int_0^t e^{-\theta(t-s)} dB_s.
    \end{equation}
The price of a risk-free zero-coupon bond can be computed as
    \begin{eqnarray*}
	P(0, T) = \mathbb{E}\left[e^{-\int_0^T r_s ds}\right] = \exp\left(-\frac{m}{\theta}T + \left(\frac{m}{\theta} - r_0\right)\frac{1-e^{-\theta T}}{\theta}\right) \times \mathbb{E}\left[e^{- \sigma \int_0^T \int_0^s e^{-\theta(s-u)} dB_u ds}\right].
    \end{eqnarray*}
The computation of the expectation of the stochastic integral is similar to the proof of Lemma 2.4 in \cite{zhou2022}. We can apply stochastic Fubini Theorem, It\^o isometry and the variance formula \eqref{var.formula} to obtain
  \begin{eqnarray*}
      \mathbb{E}\left[e^{- \sigma \int_0^T \int_0^s e^{-\theta(s-u)} dB_u ds}\right] & = & \mathbb{E}\left[e^{-\sigma \int_0^T \frac{1-e^{-\theta(T-u)}}{\theta} dB_u} \right] \\
       & = & \mathbb{E}\left[e^{-\frac{\alpha \sigma}{\theta} \int_0^T \left(1-e^{-\theta(T-u)}\right) dW_u} \right] \times \mathbb{E}\left[e^{-\frac{\sigma}{\theta} \int_0^T \left(1-e^{-\theta(T-u)}\right) dB_u^H} \right] \\
       & = & \exp\left(\frac{\alpha^2\sigma^2}{2\theta^2}\int_0^T \left(1-e^{-\theta(T-u)}\right)^2 du\right) \times \exp \left(\frac{\sigma^2\alpha_H}{2\theta^2} \int_0^T \int_0^T  \right. \\
       & & \left.\left(1-e^{-\theta(T-u)}\right)\left(1-e^{-\theta(T-v)}\right) |u-v|^{2H-2} dudv \right).
  \end{eqnarray*}
Direct evaluation of the above integrals yields equation~\eqref{bond.p0T}. 

The bond value at time $t$ is evaluated as
  \begin{eqnarray*}
	P(t, T) = \mathbb{E}\left[e^{-\int_r^T r_s ds}\big|\mathcal{F}_t\right],
  \end{eqnarray*}
where $\mathcal{F}_t$ is the filtration generated by the mfBm. Plugging the equation \eqref{rt.sol} into the above equation, we obtain
	\begin{eqnarray*}
		P(t, T) & = & \exp\left(\int_0^t r_s ds\right) \times \exp\left(-\frac{mT}{\theta} + \frac{1}{\theta}\left( \frac{m}{\theta} - r_0\right)\left(1 - e^{-\theta T}\right)\right) \\
		& & \times \ \mathbb{E}\left(\exp\left(-\sigma \int_0^T \int_0^s e^{-\theta(s-u)} dB_u ds \right) \Big\vert \mathcal{F}_t\right) \\
		& =: & I_1 \times I_2 \times I_3.
	\end{eqnarray*}
By stochastic Fubini Theorem, the term $I_3$ can be further computed as
	\begin{eqnarray*}
		I_3 & = & \mathbb{E}\left(\exp\left(-\sigma \int_0^T \frac{e^{-\theta u} - e^{-\theta T}}{\theta} e^{\theta u} dB_u\right) \Big\vert \mathcal{F}_t\right) \\ 
		& = & \mathbb{E}\left(\exp\left(-\frac{\alpha \sigma}{\theta} \int_0^T 1 - e^{-\theta(T-u)} dW_u\right)\Big\vert \mathcal{F}_t\right) \times \mathbb{E}\left(\exp\left(-\frac{\sigma}{\theta} \int_0^T 1 - e^{-\theta(T-u)} dB_u^H \right)\Big\vert \mathcal{F}_t\right) \\
		& =:& I_{31} \times I_{32}.
	\end{eqnarray*}
For the term $I_{31}$, we have

    \begin{eqnarray*}
		I_{31} & = & \exp\left(-\frac{\alpha \sigma}{\theta} \int_0^t 1 - e^{-\theta(T-u)} dW_u\right) \times \mathbb{E}\left(\exp\left(-\frac{\alpha \sigma}{\theta} \int_t^T 1 - e^{-\theta(T-u)} dW_u\right) \right) \\
          & = & \exp\left(-\frac{\alpha \sigma}{\theta} \int_0^t 1 - e^{-\theta(T-u)} dW_u\right) \times \exp\left(\frac{\alpha^2 \sigma^2}{2\theta^2} \int_t^T \left(1 - e^{-\theta(T-u)}\right)^2 du\right) \\
	   & = & \exp\left(-\frac{\alpha \sigma}{\theta} \int_0^t 1 - e^{-\theta(T-u)} dW_u\right) \times \exp\left(\frac{\alpha^2 \sigma^2}{2\theta^2} \left((T-t) - \frac{e^{-2\theta(T-t)} - 4 e^{-\theta(T-t)} + 3}{2\theta}\right)\right).
    \end{eqnarray*}
We denote $g(u):= -\frac{\sigma}{\theta} \left(1 - e^{-\theta(T-u)}\right) \mathbbm{1}_{[0, T]}(u)$ and the term $I_{32}$ can be written as 
	$$I_{32} = \mathbb{E}\left(\exp\left(\int_0^\infty g(u) dB_u^H\right)\Big\vert \mathcal{F}_t\right).$$
By Lemma \ref{lem.condexp}, we obtain the result \eqref{bond.ptT} and complete the proof.
\end{proof}

We provide a numerical illustration in Figure~\ref{fig:zcb} to show how the value of a zero-coupon bond at time 0 is influenced by different model parameters. Figure~\ref{fig:zcb} shows that the bond value does not exhibit a clear exponential decay as the bond duration increases, since the interest rate is modeled by a mean-reverting process instead of a constant. Under different volatility levels ($\sigma = 0.5$ in the left panel and $\sigma = 1$ in the right panel), the effect of the Hurst parameter $H$ on bond values is consistent across bond durations. In general, a larger Hurst parameter leads to a higher bond value. Moreover, comparing the two panels, we can see that as the volatility parameter $\sigma$ increases, the influence of the Hurst parameter on the bond values becomes more pronounced.

\begin{figure}[!ht]
    \centering
    \begin{subfigure}[b]{0.49\textwidth}
      \centering
      \includegraphics[width=\textwidth]{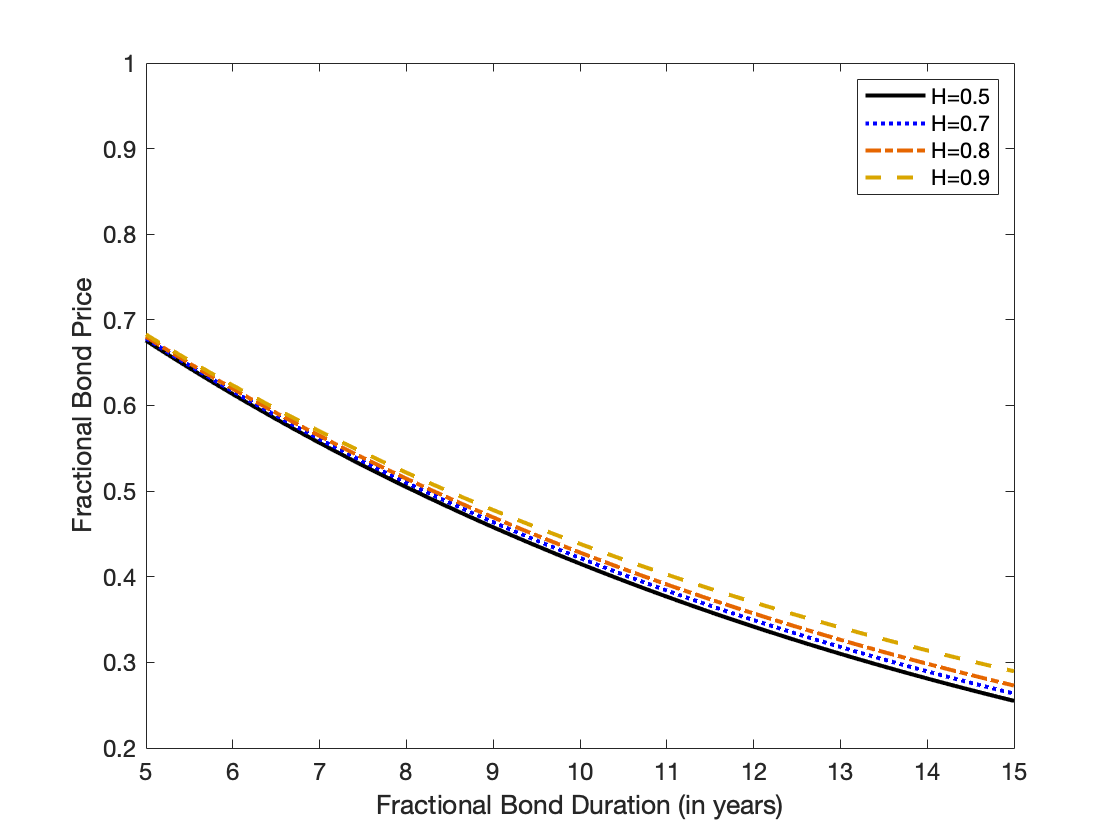}
      \caption{$\sigma = 0.5$}
    \end{subfigure}
    \begin{subfigure}[b]{0.49\textwidth}
      \centering
      \includegraphics[width=\textwidth]{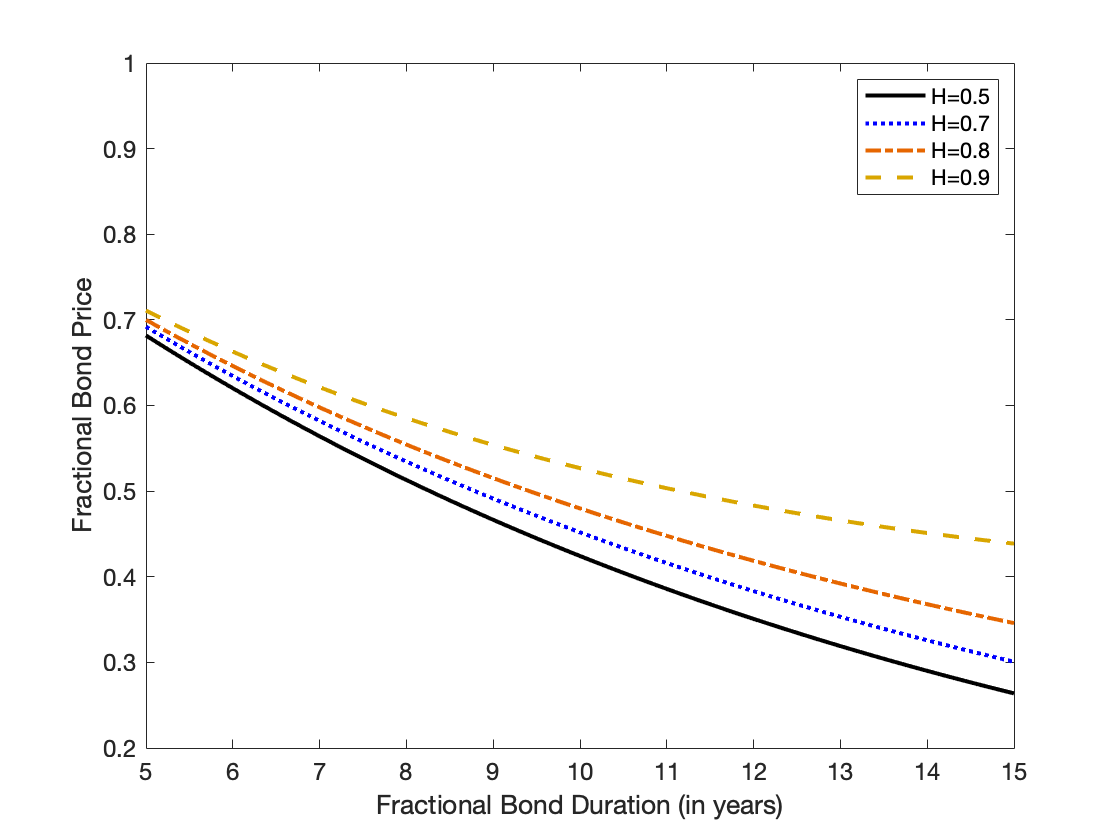}
      \caption{$\sigma = 1$}
    \end{subfigure}\caption{An illustration of zero-coupon bond value at time 0 under different values of $H$ and $\sigma$, where $\alpha = 0$, $m = 10$, and $\theta = 1$.}
    \label{fig:zcb}
\end{figure}

\subsection{Pricing mortality-linked securities}
\label{MLS.pricing}

Mortality bonds are a type of mortality-linked security (MLS) that enables insurance companies and pension plans to transfer catastrophic mortality risk to capital markets. For example, Swiss Re introduced Vita Capital bonds in 2003 to hedge against increased mortality rates in the United States, United Kingdom, Canada, and Germany. The payoff structure of these bonds is tied to a specific mortality index, such as a national death rate. If this index exceeds a predetermined threshold, the principal repayment of the bond is reduced proportionally, allowing issuers to mitigate losses from unexpected mortality events such as pandemics or natural disasters. In return, investors accept this risk in exchange for yields that are typically higher than those offered by standard bonds.

Consider an MLS issued at time $0$ with scheduled payments at $t_k \in \{t_1, t_2, \ldots, t_K\}$, where $t_K = T$ is the maturity time. At each time point $t_k$, the bond payment $\phi_{t_k}$ depends on a mortality-linked quantity. It is assumed that for all $k = 1, 2, \ldots, K$, the payment function $\phi_{t_k}$ is measurable with respect to the augmented filtration $\mathcal{F}_{t_k}^r$ generated by the stochastic processes $\{W_{i,u},\ 0 < u < t_k\}$ and $\{B^{H_i}_{i,u},\ 0 < u < t_k\}$ for $i = 1, 2$. 

For a mortality bond structured similarly to Vita bonds, assume the face value is $F$ and the annual coupon rate is $c$, payable $m$ times per year. The payment $\phi_{t}$ at times $t \in \{t_1, t_2, \ldots, t_K\}$ is given by
\begin{align}\label{mort.bond.payment}
\phi_{t} = \begin{cases} \frac{cF}{m} & \mbox{if } t=t_1, ..., t_{K-1} \\
\frac{cF}{m} + F(1-{\rm PRF}) & \mbox{if } t=t_K \end{cases},
\end{align}
where ${\rm PRF}$ is the principal reduction factor that determines the proportion of the principal repaid to investors at maturity. The PRF depends on the realized value of a mortality index $\tilde{\mu}_t$ for $t \in \{t_1, t_2, \ldots, t_K\}$, and is defined as
$$
{\rm PRF}(\tilde{\mu}) = \min\left(1, \; \sum_{t= t_1}^{t_K}\frac{(\tilde{\mu}_{t} - a)_+ - (\tilde{\mu}_{t} - b)_+}{b - a} \right),
$$
where $a$ is the attachment point and $b$ is the exhaustion point. If $\tilde{\mu}_t < a$ for all $t$, then ${\rm PRF} = 0$ and the principal is fully repaid. In contrast, if $\tilde{\mu}_t > b$ for any $t$, then ${\rm PRF} = 1$, meaning that the principal is entirely lost.

Regarding the mortality index linked to the bond, we present the following considerations. First, it should be designed to accurately reflect the mortality experience of the population over the term of the bond. Several specifications for $\tilde{\mu}_t$ can be considered:
\begin{itemize}
    \item $\tilde{\mu}_t$ is set to the mortality rate observed at time $t$,
    \item $\tilde{\mu}_t$ is set to the average mortality rate observed between times $t-1$ and $t$, or
    \item $\tilde{\mu}_t$ is set to the maximum mortality rate observed between times $t-1$ and $t$.
\end{itemize}

Second, the choice of stochastic model plays an important role in determining the mortality index, which is an essential part in the design of the MLS. In this paper, we assume that the mortality rate follows the SDE \eqref{int.mort.sde}, which highlights the LRD property observed in human mortality data. Existing literature such as \cite{zhou2022} has studied how the Hurst parameter affects the distribution of lifetime related random variables. Since our focus is on mortality bonds, it is also important to understand how extreme values of mortality respond to changes in the Hurst parameter. Figure~\ref{fig:mort.model} illustrates the impact of Hurst parameter $H_2$ on the maximum mortality rate over the bond term. A higher value of $H_2$ leads to a higher exceedance probability when the threshold mortality rate is set high. This suggests that stronger LRD would increase the likelihood of extreme mortality rates, which further affects the design and pricing of mortality bonds.

\begin{figure}[!ht]
       \centering
	   \begin{subfigure}[b]{0.49\textwidth}
          \includegraphics[width=\textwidth]{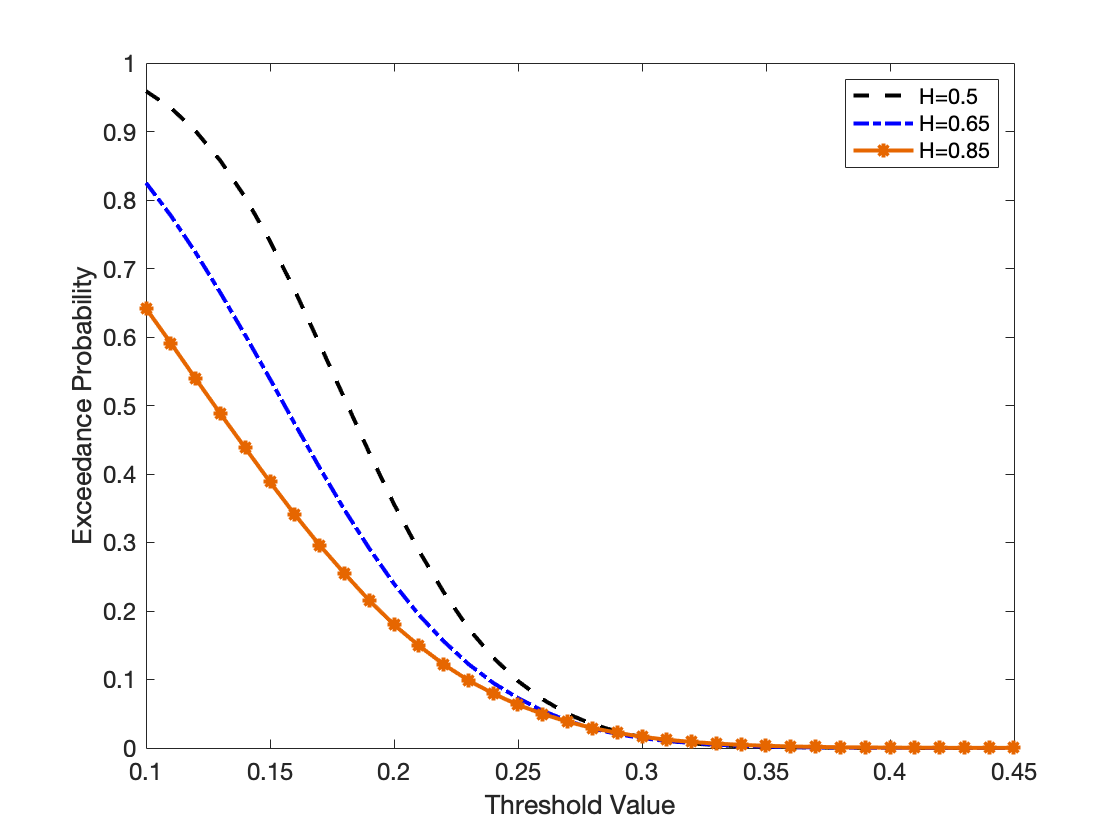}
          \caption{Threshold value: 0.1 to 0.45}
	   \end{subfigure}
	   \begin{subfigure}[b]{0.49\textwidth}
		   \includegraphics[width=\textwidth]{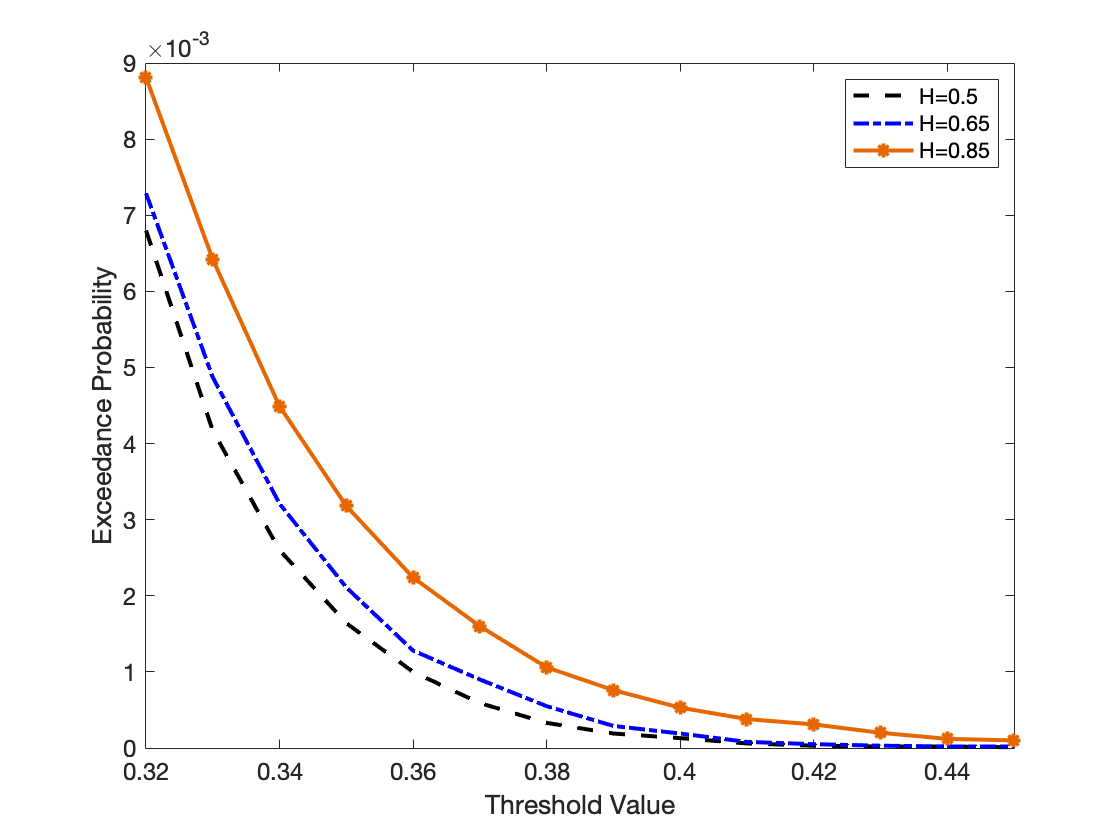}
          \caption{Threshold value: 0.32 to 0.45}
	   \end{subfigure}
       \caption{Probability that the maximum excess mortality $\tilde\mu = \max_{0 \leq t \leq T} \mu_t$ exceeds a given threshold under different values of $H_2$. Results are based on $10^5$ simulated trajectories of $\mu_t$ from the SDE with parameters $\alpha = 0$, $\sigma = 0.116724$, $m = 0.008818679$, $\theta = 1.029601$, and $\mu_0 = -0.085646188$.}
       \label{fig:mort.model}
\end{figure}

Based on the payout structure given in equation \eqref{mort.bond.payment} , the time-$t$ value of future bond payments can be expressed as:
\begin{align*}
P_t^{\rm MB} &= \mathbb{E}^{\mathbb{Q}}\left[ \sum_{t < t_k \leq T} \phi_{t_k} e^{-\int_t^{t_k} r_s ds}  \Big|\mathcal{F}_t^r \right]\\
&= \sum_{t < t_k \leq T}  \frac{cF}{m} \mathbb{E}^{\mathbb{Q}}\left[ e^{-\int_t^{t_k} r_s ds}  \Big|\mathcal{F}_t^r \right] + \mathbb{E}^{\mathbb{Q}}\left[ F(1-{\rm PRF(\tilde \mu)}) e^{-\int_t^{T} r_s ds}  \Big|\mathcal{F}_t^r \right] \\ 
&= \sum_{t < t_k \leq T}  \frac{cF}{m} P(t, t_k) + F P(t, T) - F\mathbb{E}^{\mathbb{Q}}\left[ {\rm PRF}(\tilde \mu) \times e^{-\int_t^T r_s ds}  \Big|\mathcal{F}_t^r \right],
\end{align*}
where $P(t, t_k)$, for $t < t_k \leq T$, represents the time-$t$ price of a risk-free zero-coupon bond maturing at time $t_k$. The explicit formula for $P(t, t_k)$ has been derived in Theorem~\ref{risk.free.bond}.

Finally, the price of the mortality bond at time $t = 0$, with an annual coupon rate of $c$ payable $m$ times per year, is given by
\begin{equation}\label{price.mort.bond}
P_0^{\rm MB}(c) = \frac{cF}{m} \sum_{t_k=t_1}^{t_K} P(0, t_k) + F P(0, T) - F \mathbb{E}^{\mathbb{Q}}\left[ {\rm PRF}(\tilde{\mu}) \times e^{-\int_0^T r_s ds} \, \Big| \mathcal{F}_0^r \right],
\end{equation}
where the term $\mathbb{E}^{\mathbb{Q}}\left[ {\rm PRF}(\tilde{\mu}) \times e^{-\int_0^T r_s ds} \Big| \mathcal{F}_0^r \right]$ represents the expected value of the discounted principal reduction and can be evaluated numerically using simulations.

\section{Parameter estimation} \label{sec:ParaEst}

To correctly apply the above analysis results on the MLS, it is important to calibrate model parameters in both the physical measure and the pricing measure. As the mortality-linked insurance market is believed to be incomplete, we first describe a data-driven calibration process under real-world scenarios based on the asymptotic results of the SDE. Then we propose a Monte Carlo simulation algorithm combined with an optimization procedure to obtain the risk premium parameters based on the limited price information regarding a catastrophe mortality bond available in the market.

\subsection{Parameter estimation of a general SDE driven by mfBm}\label{pe.sde}

In this subsection, we consider the parameter estimation of a general SDE driven by mfBm, 
   $$d X_t  =  (m - \theta X_t) dt + \sigma (\alpha dW_t + d B_t^H),$$
where $X_0 = x_0$. Suppose that we have one observed trajectory of $X_t$ at equidistant time points $t=0, \frac{1}{n}, \frac{2}{n}, \ldots, \frac{[nT]}{n}$. The parameters $m, \theta, \sigma, \alpha, H$ are unknown. 

We employ the well-known rescaled range (R/S) analysis (see, for example, \cite{andrew91}, \cite{mandelbrot72}, \cite{mandelbrot69},  \cite{xiao23}) to estimate the Hurst parameter $H$. Namely, for a long memory series of length $N$, $\left\{X_{\frac{i}{n}}, i=0, 1, \ldots, N\right\}$, the rescaled range can be calculated as follows:
    $$Y_j = X_{\frac{j}{n}} - \frac{1}{N} \sum_{i=1}^N X_{\frac{j}{n}}, \ j=1, 2, \ldots, N; \ Z_k = \sum_{j=1}^k Y_j, k = 1, 2, \ldots, N.$$
The range $R$ and rescaled range $R/S(N)$ can be computed by
    $$R(N) = \max(Z_1, \ldots, Z_N) - \min(Z_1, \ldots, Z_N), \ R/S(N) = \frac{R(N)}{\sqrt{\frac{1}{N}\sum_{j=1}^N Y_j^2}}.$$
Then the Hurst parameter $H$ can be estimated as 0.5 plus the slope of $R/S(N) - \mathbb{E}[R/S(N)]$ against $N$. The formula for approximating the expectation of $R/S(N)$ can be found in (3.9) of \cite{annis76}.

To estimate the volatility parameter $\sigma$ and weight parameter $\alpha$, we define the power variations of the stochastic process $X = \{X_t, t \geq 0\}$ as
   $$V^n(X)_t = \sum_{i=1}^{[nt]}\left|\Delta X_{\frac{i-1}{n}}\right|^2 = \sum_{i=1}^{[nt]}\left|X_{\frac{i}{n}} - X_{\frac{i-1}{n}}\right|^2,$$
   $$U^n(X)_t = \sum_{i=1}^{[nt]}\left|\tilde\Delta X_{\frac{i-1}{n}}\right|^2 = \sum_{i=1}^{[nt]}\left|X_{\frac{i+1}{n}} - X_{\frac{i-1}{n}}\right|^2.$$
We may find that at time $t=T$,
   $$V^n(X)_T = \sigma^2 \left(\alpha^2 V^n(W)_T + V^n(B^H)_T + 2\alpha \sum_{i=1}^{[nT]}\left|\Delta W_{\frac{i-1}{n}}\right|\left|\Delta B^H_{\frac{i-1}{n}}\right|\right) + o(n^{-1}),$$
   $$U^n(X)_T = \sigma^2 \left(\alpha^2 U^n(W)_T + U^n(B^H)_T + 2\alpha \sum_{i=1}^{[nT]}\left|\tilde\Delta W_{\frac{i-1}{n}}\right|\left|\tilde\Delta B^H_{\frac{i-1}{n}}\right|\right) + o(n^{-1}).$$
From the ergodic theorem, we have the following asymptotic results (see \cite{hu19} and the references therein),
$$ \lim_{n \to \infty} V^n(W)_T =T, \ \lim_{n \to \infty} n^{2H-1} V^n(B^H)_T = T ,$$ 
and $$\lim_{n \to \infty} n^{H-\frac{1}{2}}\sum_{i=1}^{[nT]}\left|\Delta W_{\frac{i-1}{n}}\right|\left|\Delta B^H_{\frac{i-1}{n}}\right| = T,$$
in probability. Similarly,
$$ \lim_{n \to \infty} U^n(W)_T = \mathbb{E}\left[|W_2|^2\right] = 2T, \ \lim_{n \to \infty} n^{2H-1} U^n(B^H)_T = \mathbb{E}\left[\left|B_2^H\right|^2\right] = 2^{2H} T,$$ 
and $$\lim_{n \to \infty} n^{H-\frac{1}{2}}\sum_{i=1}^{[nT]}\left|\tilde\Delta W_{\frac{i-1}{n}}\right|\left|\tilde\Delta B^H_{\frac{i-1}{n}}\right| = 2^{H+\frac{1}{2}}T,$$
in probability. Taking into account of $H > \frac{1}{2}$, we have
   $$\lim_{n \to \infty} V^n(X)_T = \frac{1}{2} \lim_{n \to \infty} U^n(X)_T = \alpha^2 \sigma^2 T,$$
and
  $$\lim_{n \to \infty} n^{H-\frac{1}{2}}\left(-2V^n(X)_T + U^n(X)_T\right) = 2\alpha\sigma^2 T\left(2^{H+\frac{1}{2}}-2\right),$$
in probability. Based on the above results, we can construct the consistent estimators for the parameters $\alpha$ and $\sigma$:
    \begin{equation}\label{alpha.pe}
        \hat\alpha = \frac{2 \left(2^{H+\frac{1}{2}}-2\right) V^n(X)_T}{n^{\hat H-\frac{1}{2}}\left(U^n(X)_T - 2V^n(X)_T\right)},
    \end{equation}
    \begin{equation}\label{sigma.pe}
       \hat\sigma = \sqrt{\frac{V^n(X)_T}{\hat\alpha^2 T + T n^{1-2\hat H} + 2 \hat\alpha T n^{\frac{1}{2} - \hat H}}}.
    \end{equation}
The parameters $m, \theta$ in the drift part can be estimated using the least squares method (see, for example, \cite{hu19} and \cite{Tindel2014} for the idea). Namely, we would like to minimize the function $$L(m, \theta) = \sum_{i=1}^{[nT]} \left|X_{\frac{i}{n}} - X_{\frac{i-1}{n}} - \frac{1}{n}\left(m - \theta X_{\frac{i-1}{n}}\right)\right|^2.$$ Then the least squares estimators $\hat m$ and $\hat \theta$ can be defined by $(\hat m, \hat \theta) = {\rm argmin} \ L(m, \theta)$, and we can derive the explicit expressions for the estimators:
    \begin{equation}\label{m.pe}
        \hat m = n\frac{\left(\sum_{i=1}^{[nT]} \Delta X_{\frac{i-1}{n}} \right) \left(\sum_{i=1}^{[nT]} X^2_{\frac{i-1}{n}}\right) - \left(\sum_{i=1}^{[nT]} \Delta X_{\frac{i-1}{n}} \cdot X_{\frac{i-1}{n}} \right) \left(\sum_{i=1}^{[nT]} X_{\frac{i-1}{n}}\right)}{ [nT]\sum_{i=1}^{[nT]} X^2_{\frac{i-1}{n}} - \left(\sum_{i=1}^{[nT]} X_{\frac{i-1}{n}}\right)^2},
    \end{equation}
    \begin{equation}\label{theta.pe}
      \hat \theta = n\frac{\left(\sum_{i=1}^{[nT]} \Delta X_{\frac{i-1}{n}} \right) \left(\sum_{i=1}^{[nT]} X_{\frac{i-1}{n}} \right) - [nT]\sum_{i=1}^{[nT]}\Delta X_{\frac{i-1}{n}} \cdot X_{\frac{i-1}{n}}}{[nT]\sum_{i=1}^{[nT]} X^2_{\frac{i-1}{n}}- \left(\sum_{i=1}^{[nT]} X_{\frac{i-1}{n}}\right)^2}.
    \end{equation}
	
Alternatively, the parameters $m, \theta$ can be estimated by the ergodic-type estimators that are consistent (see \cite{hu19}, \cite{xiao23}). Namely, the mean reversion rate $\theta$ can be estimated by 
	 $$\hat\theta = 
     \left(\frac{[nT]\sum_{i=1}^{[nT]} X_{\frac{i}{n}}^2 - \left(\sum_{i=1}^{[nT]} X_{\frac{i}{n}}\right)^2}{[nT]^2 \hat\sigma^2 \hat H \Gamma(2\hat H)}\right)^{-\frac{1}{2 \hat H}}.$$
 The quantity $\frac{m}{\theta}$ is known as the long term mean of $X_t$ and it can be determined from an external resource (such as expert opinions or actuarial judgment on the long-term trend of $X_t$). Based on the estimation of the long-term mean, we can estimate the parameter $m$.

\subsection{Model calibration for the bivariate process $(r_t, \mu_t)$}\label{pe.biv}
		 
  In this section, we apply the methods described in Section~\ref{pe.sde} to estimate the eleven parameters in the bivariate process $(r_t, \mu_t)$ given by \eqref{int.mort.sde}. The sequential algorithm is described in the following six steps.
  
  \noindent Step 1: Apply the rescaled range analysis method to estimate the Hurst parameter $H_1$ for the interest rate process $r_t$. \\
  Step 2: Apply the formulas in Section~\ref{pe.sde} to estimate the parameter $\alpha_1$ and the volatility parameter $\sigma_1$ for the process $r_t$.\\
  Step 3: Estimate the parameters $m_1$ and $\theta_1$ for the process $r_t$ based on ergodic-type estimators. The motivation is to align the estimated mean reverting level of the SDE with the long-term mean, which is well known in practice.\\
  Step 4: Repeat the above steps 1 $\sim$ 3 to estimate the parameters $H_2$, $\alpha_2$, and $\sigma_2$ for the excess mortality process $\mu_t$.\\
  Step 5: Estimate the parameters $m_2$ and $\theta_2$ based on the least squares method, i.e., formulas \eqref{m.pe} and \eqref{theta.pe}, for the excess mortality process $\mu_t$. \\
  Step 6: Finally, we estimate the parameter $\rho$ using the formula
	$$\hat\rho = \frac{n \sum_{i=1}^{[nT]} (e_i^{(r)} - \bar e^{(r)})(e_i^{(\mu)} - \bar e^{(\mu)})}{\hat\sigma_1 \hat\sigma_2 \hat\alpha_1 \hat\alpha_2},$$
    where the sequences $e_i^{(r)}$ and $e_i^{(\mu)}$ can be derived from 
	$$e_i^{(r)} = r_{\frac{i}{n}} - r_{\frac{i-1}{n}} - \left(\hat m_1 - \hat\theta_1 r_{\frac{i-1}{n}}\right) \cdot \frac{1}{n},$$
	$$e_i^{(\mu)} = r_{\frac{i}{n}} - r_{\frac{i-1}{n}} - \left(\hat m_2 - \hat\theta_2 r_{\frac{i-1}{n}}\right) \cdot \frac{1}{n},$$
    with $\bar e^{(r)} = \frac{1}{[nT]} \sum_{i=1}^{[nT]} e_i^{(r)}, \ \bar e^{(\mu)} = \frac{1}{[nT]} \sum_{i=1}^{[nT]} e_i^{(\mu)}$. \\
	
\subsection{Calibration of pricing parameters}\label{pe.qsde}

The previous section addressed the parameter estimation problem for the excess mortality and interest rate models under the physical measure. In this section, we present a general calibration procedure for estimating the pricing parameters under a selected risk-neutral pricing measure. The proposed methodology will be implemented later using market data from the Swiss Re Vita VI bond.

We assume that the mortality bond is a par value bond, which means $P_0^{\rm MB}(c) = F$ and the risk premium is reflected in the coupon rate $c$. In other words, the coupon rate $c$ should be the risk-free rate at bond issuance plus a premium spread such that the bond price calculated from equation \eqref{price.mort.bond} equals $F$. Setting $P_0^{\rm MB} = F$, we can derive the formula of $c$ from equation \eqref{price.mort.bond} to obtain
\begin{equation} \label{eq:coupon}
c = \frac{1 - P(0, T) + \mathbb{E}^{\mathbb{Q}}\left[{\rm PRF}(\tilde \mu) \times e^{-\int_0^T r_s ds} \big|\mathcal{F}_0^r\right]}{\frac{1}{m}\sum_{0 < t_k \leq T}P(0, t_k)}.
\end{equation}

The mortality-linked security (MLS) traded in the market generally provides information about the coupon rate $c_{\rm obs}$, expected loss $EL_{\rm obs}$, probability of attachment $PA_{\rm obs}$, and premium spread $PS_{\rm obs}$. The general principle is to match the mathematical equations derived from the stochastic models with the market information. The calibration process for the MLS involves numerical solutions for the following parameters: attachment point $a$, exhaustion point $b$, risk premiums $\gamma_1, \gamma_2, \eta_1, \eta_2$. In general, we can follow the algorithm below to estimate these parameters. \\

\noindent Step 1: Obtain simulated mortality rate trajectories over the MLS term from the mortality model under physical measure $\mathbb{P}$. For our proposed model, we can simulate the excess mortality rates from the model \eqref{int.mort.sde} and recover the mortality rates by adding back the selected baseline mortality rates. \\

\noindent Step 2: Calculate the mortality index $\tilde\mu$ based on the MLS design as discussed in Section \ref{MLS.pricing} for all simulated trajectories obtained in Step 1 and derive the estimated distribution of $\tilde\mu$. If the probability of attachment information is available on the market, the attachment point $a$ can be estimated by the $100(1-PA_{\rm obs})$th percentile of the distribution of $\tilde\mu$. If the expected loss information is available, we can obtain an estimation of the exhaustion point $b$ through a selected optimization method such that
   $$EL_{\rm obs} = \mathbb{E}[{\rm PRF}(\tilde \mu)].$$

\noindent Step 3: For candidate values of the risk premium parameters, simulations are performed for the interest rate and the excess mortality rate under the pricing measure $\mathbb{Q}$ so that the right-hand side of equation \eqref{eq:coupon} can be numerically approximated. Then an optimization approach can be implemented to minimize the absolute differences between the observed coupon rates and the approximated coupon rates from numerical simulations. The optimization process is similar if the premium spread information is available, as it is essentially the difference between the coupon rate and the risk-free interest rate.

The above general methodology can be applied for many types of MLS, but the specific formulas and selection of the optimization procedure will be customized and provided for a real-world mortality bond in Section~\ref{sec:Qcali}.

\section{Numerical Illustration} \label{sec:NumIll}

In this section, we present a numerical illustration of the proposed model and pricing framework presented in Sections \ref{sec:ModelDescr} and \ref{sec:PricFram}, using real-world data applied to a catastrophe mortality bond. We begin by presenting the calibration results obtained through the methodologies and procedures described in Section \ref{sec:ParaEst}. We then examine the pricing outcomes and risk characteristics of the bond under the calibrated model. Lastly, we conduct an extensive sensitivity analysis to investigate how changes in model assumptions and pricing parameters influence the bond's coupon rate and risk profile.

\subsection{Data description} \label{sec:NumIll_Data}

The mortality and interest rate data used in our numerical illustration are from 2015 to 2024, with a total of 521 weekly observations. The data are collected from the following sources:

\begin{itemize}
    \item {\it Human Mortality Database}\footnote{HMD. Human Mortality Database. Max Planck Institute for Demographic Research (Germany), University of California, Berkeley (USA), and French Institute for Demographic Studies (France). Available at \textit{www.mortality.org} (data downloaded on March 26, 2025).}: The mortality data are taken from the Short-Term Mortality Fluctuations (STMF) series for the US population and consist of weekly, all-age, unisex death counts and exposures. The excess mortality rate at week $t$ is computed as
    \[
    \mu_t = \frac{D_t}{E_t} - \mathbb{E}\left[\frac{D_t}{E_t}\right],
    \]
    where $D_t$ is the observed number of deaths and $E_t$ is the corresponding exposure in week $t$. To remove seasonal variation, the expected death rate $\mathbb{E}[D_t/E_t]$ is estimated using week-specific averages over the years 2015-2019 (i.e., the pre-COVID years of our dataset). That is, the expected rate for the $m$-th week of the year is calculated as the average observed rate for the $m$-th week across those years.

    \item {\it Federal Reserve Economic Data (FRED)}\footnote{Board of Governors of the Federal Reserve System (US), 3-Month Treasury Bill Secondary Market Rate, Discount Basis [WTB3MS], retrieved from FRED, Federal Reserve Bank of St. Louis; https://fred.stlouisfed.org/series/WTB3MS, March 26, 2025.}: The interest rate data consist of weekly 3-month Treasury bill rates from the WTB3MS series. These rates are used directly as $r_t$ in the model without further adjustment.

\end{itemize}

\subsection{Calibration results}

This section presents the calibration results for our proposed model \eqref{int.mort.sde} and its risk neutral version \eqref{int.mort.qsde}. While a general overview of the calibration procedure for the risk premium parameters is provided in Section~\ref{pe.qsde}, this section offers a more detailed explanation of the procedure customized for the data described in Section \ref{sec:NumIll_Data}.

\subsubsection{Calibration under measure $\mathbb{P}$ for model parameters}
We apply the methods described in Section~\ref{pe.sde} and Section~\ref{pe.biv} to calibrate the model \eqref{int.mort.sde} using the collected data on the mortality and interest rates. The calibration results are summarized in Table \ref{para.est1}. The top row reports parameter estimates for the interest rate process $r_t$, and the bottom row reports parameter estimates for the excess mortality process $\mu_t$. The Hurst parameters $H$ for both $r_t$ and $\mu_t$ are clearly greater than $0.5$, showing that both the interest rate process and the excess mortality rate process possess a LRD property. The estimated instantaneous correlation between $r_t$ and $\mu_t$ is $\rho = -0.29265$. The negative correlation is expected, as empirical observations show that interest rates and mortality rates have tended to move in opposite directions in recent years, particularly during the COVID pandemic. While this correlation does not imply a causal relationship, both variables may be influenced by common external factors.

\begin{table}[!ht]
    \centering
    \caption{Estimated parameters of the proposed model calibrated to weekly mortality and interest rate data from 2015 to 2024.}
    \begin{tabular}{ccccccc}
	\hline
	\qquad\quad\quad\quad & $H$ & $\alpha$ & $\sigma$ & $m$ & $\theta$ & $\rho$ \\
	\hline
	$r_t$ & $0.85957$ & $0.24815$ & $1.24565$ & $2.26377$ & $0.54157$ & \multirow{2}{*}{$-0.29265$} \\
	$\mu_t$ & $0.78416$ & $0.32636$ & $0.00286$ & $0.00068$ & $1.17364$ &  \\
        \hline
    \end{tabular}
    \label{para.est1}
\end{table}


Moreover, both $\alpha$ estimates, which represent the weights on the Brownian motion components, are substantially below 1. This suggests that the fractional Brownian motion components play a more dominant role in driving the dynamics of both the interest rate and excess mortality rate over the observation period. As discussed earlier, the presence of long memory, captured by the fractional Brownian motion, increases the instantaneous variances of the processes. This inflation reduces the magnitude of the instantaneous correlation, keeping it below 0.29265 over small time intervals. Finally, the estimated mean reversion rate for the mortality process exceeds that of the interest rate, implying that mortality tends to revert to its long-term mean more quickly than the interest rate.

\subsubsection{Calibration under measure $\mathbb{Q}$ for pricing parameters}
\label{sec:Qcali}

We use real market data to calibrate the pricing parameters under the risk-neutral measure $\mathbb{Q}$, as described in Theorem~\ref{prob.measure.q}. In particular, we calibrate the pricing parameters to match the price of the Vita Capital VI bond, issued in July 2021. The relevant pricing information of this bond, collected from Artemis\footnote{Source: \url{https://www.artemis.bm/deal-directory/vita-capital-vi-limited-series-2021-1/}}, is reported in Table~\ref{tab:VitaData}. Based on the available market information, we assume that the Vita VI bond is a 5-year par value bond that makes annual payments following equation \eqref{mort.bond.payment} at a specified coupon rate of $3\%$. The mortality index that determines the bond's principle reduction factor is the average mortality rate observed over each year during the 5-year term. 

\begin{table}[h]
   \centering
   \caption{Pricing information of the Vita Capital VI bond.}
   \begin{tabular}{c|c|c|c}
      \hline
      Term & Probability of First Loss & Expected Loss & Coupon Rate \\
      \hline
      5 Years & 0.75\% & 1.06\% & 3\% \\
      \hline
   \end{tabular}
   \label{tab:VitaData}
\end{table}

The general methodology in Section~\ref{pe.qsde} can be applied to find other unknown parameters of this specific mortality bond. Step 1 and Step 2 in Section~\ref{pe.qsde} can be implemented to determine the attachment point and the exhaustion point. To match the coupon rate, the evaluation of equation \eqref{price.mort.bond} is not straightforward and it may not be easy to find a closed formula due to the challenge of deriving the conditional expectation of a functional that involves extremal functionals with respect to fBm. However, we will perform the simulation-based optimization procedure as described in Step 3 of Section~\ref{pe.qsde}.

As we have limited data information, to reduce the dimensionality of the parameter's space for optimization, we mandate that the risk premiums for the fBm part are zero, that is, $\eta_1 = \eta_2 = 0$. To determine the risk premium $\gamma_1$, we can implement an optimization procedure to match the price of the zero coupon bond under the measure $\mathbb{Q}$ with the one determined by the target interest rate $i$. That is, the optimal value of the risk premium $\gamma_1$ is given by
  $${\rm argmin}_{\gamma_1} |P(0, T)_{m_1'} - (1+i)^{-5}|,$$
where $P(0, T)_{m_1'}$ is the price of the zero coupon bond from Theorem~\ref{risk.free.bond} under the probability measure $\mathbb{Q}$ with the parameter $m_1' = m_1 + \alpha_1 \sigma_1 \gamma_1 + \sigma_1 \eta_1$ given in Theorem~\ref{prob.measure.q}. The interest rate $i$ measures the expected yield rate of a 5-year corporate bond for pricing purposes. We choose $i=2.57\%$, which is the monthly average AAA bond (annualized) spread when the mortality bond was issued. We adopt the fmincon optimization method in Matlab to obtain the optimal value of the parameter $\gamma_1$ as $\hat\gamma_1 = 3.8701$.

The attachment point can be found numerically by finding a value $a$ such that the probability of attachment (i.e., the probability of first loss) is 1.06\%,
  $$P\left({\rm PRF}(\tilde{\mu}) > 0\right) = 1.06\%.$$
Note that ${\rm PRF}(\tilde{\mu}) > 0$ if and only if $\{\tilde{\mu}_{t_k} > a | k = 1, \ldots, 5\} \neq \emptyset$.
We can simulate 10,000 trajectories of excess mortality according to the SDE under the pricing measure $\mathbb{Q}$ and add the baseline mortality rate (i.e., the expected death rate $\mathbb{E}[D_t/E_t]$) to the simulated excess mortality. In this way, we obtain the 10,000 simulated weekly mortality rates $\mu_t$. The principal reduction factor (PRF) is then determined as
  $${\rm PRF}(\tilde{\mu}) = \min\left(1, \sum_{t_k= 1}^5 \frac{(\tilde{\mu}_{t_k} - a)_+ - (\tilde{\mu}_{t_k} - b)_+}{b - a} \right),$$
where $\tilde{\mu}_{t_k}$ is the average mortality rate for year $t_k$, determined by the simulated weekly mortality rates and computed as $\tilde{\mu}_{t_k} =\frac{\sum_{m=1}^{52}\mu_{52\times t_k +m}}{52}$.
The $98.94$th percentile of all values $\{\tilde{\mu}_{t_k}: k=1,\ldots,5\}$ will be the estimated value of the attachment point $\hat a$. 

Given the expected loss data information, we can set up the equation
  $$\mathbb{E}\left[{\rm PRF} (\tilde{\mu})\Big|{\rm PRF}(\tilde{\mu}) > 0 \right] \times P\left({\rm PRF}(\tilde{\mu}) > 0\right) = 0.75\%.$$
The exhaustion point $b$ can be estimated by solving the above equation via minimizing the function
   $$f(b) = \left|\frac{1}{N}\sum_{i=1}^N \min\left(1, \sum_{t_k= 1}^T \frac{(\tilde{\mu}_{i, t_k} - a)_+ - (\tilde{\mu}_{i, t_k} - b)_+}{b - a} \right)  - \frac{0.75\%}{1.06\%}\right|,$$
where $\tilde{\mu}_{i, t_k}$ is the $i$th simulated value of $\tilde{\mu}_{t_k}$ that satisfies $\tilde{\mu}_{i, t_k} > a$, and $N$ is the number of simulated trajectories such that the attachment point is triggered. 

Our last step is to find the value of the risk premium $\gamma_2$ such that the coupon rate 3\% satisfies the pricing equation \eqref{eq:coupon}. To achieve this purpose, we implement the grid search method for the target value $\gamma_2$ in the interval $[0, 2]$ with step size $0.001$, such that
$$3\% = \frac{1 - P(0, T) + \frac{1}{N} \sum_{i=1}^N \min\left(1, \sum_{t_k= 1}^T \frac{(\tilde{\mu}_{i, t_k} - a)_+ - (\tilde{\mu}_{i, t_k} - b)_+}{b - a} \right) \times e^{- \sum_{t_k=1}^5 \sum_{m=1}^{52} r_{i, 52t_k + m}}}{\frac{1}{m}\sum_{0 < t_k \leq T}P(0, t_k)}$$
holds numerically, where $r_{i, 52t_k + m}$ are the simulated short rates under pricing measure $\mathbb{Q}$. The optimal value for the risk premium $\gamma_2$ is finally found to be $\hat\gamma_2 = 1.0620$.

\subsection{Baseline analysis}

After calibrating the model and pricing parameters to real mortality data and market information, we now analyze a $5$-year catastrophe mortality bond structured similarly to the Vita VI bond. We set the face value to $F = 100$ and assume that the bond makes annual payments according to the formula \eqref{mort.bond.payment}. The mortality index that determines the bond’s principal reduction factor is defined as the yearly average mortality rate over the $5$-year term. We use simulated paths of mortality and interest rates under measure $\mathbb{Q}$ to determine the coupon rate, and under measure $\mathbb{P}$ to compute the present value of the bond's payouts at issuance (i.e., time 0).

\subsubsection{Simulation of mortality and interest rates}

We begin by presenting the forecast results generated by our calibrated model under the risk-neutral measure $\mathbb{Q}$. These simulated paths are used to calculate the fair coupon rate $c$ based on Equation~\eqref{eq:coupon}. The simulated mortality rates are computed as the sum of the excess mortality rates generated by our model and the expected death rate $\mathbb{E}[D_t/E_t]$. Recall from Section~\ref{sec:NumIll_Data} that $\mathbb{E}[D_t/E_t]$ is estimated as the average of the weekly mortality rates observed during the years 2015-2019 (i.e., the pre-COVID years of our dataset).

\begin{figure}[htbp]
    \centering
    \includegraphics[width=0.95\textwidth]{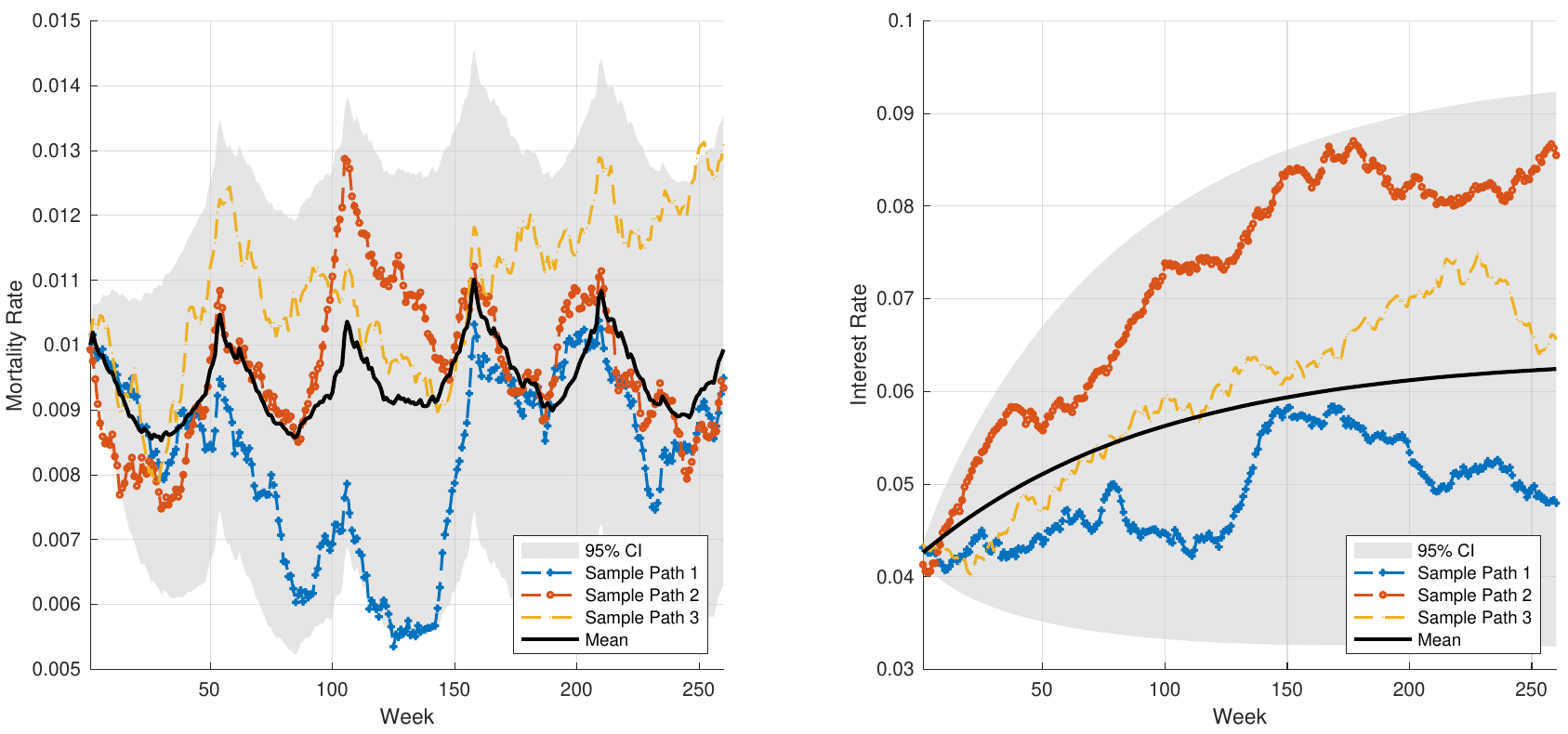}
    \caption{Simulated paths of weekly mortality rates (left panel) and interest rates (right panel) with 95\% confidence intervals under $\mathbb{Q}$-measure.}
    \label{fig:sim_paths_mortality_interest}
\end{figure}

Figure~\ref{fig:sim_paths_mortality_interest} shows the simulated weekly paths of mortality and interest rates over a 250-week period. Each plot includes three representative simulation paths, along with the estimated mean and 95\% confidence intervals. The simulations are generated using our calibrated model under the pricing measure $\mathbb{Q}$. The left panel shows the mortality rates fluctuating around the seasonal means with occasional upward spikes, while the right panel shows simulated interest rates that exhibit strong mean-reverting behavior within a plausible range around 3\% to 9\%, consistent with historical interest rate behavior. The gradual widening of the confidence intervals in both panels highlights the persistence and increasing uncertainty over time, a feature well captured by our mfBm model structure. 

\subsubsection{Principal repayment}

We now investigate the payments of the catastrophe mortality bond using our calibrated model under the physical measure $\mathbb{P}$. Specifically, we analyze the distribution of principal repayment and total payouts based on simulated mortality and interest rates. Figure~\ref{fig:prf_histogram} presents the distribution of simulated principal repayment values at time 0, computed under the physical measure $\mathbb{P}$. Each value represents the time-0 present value of the principal repayment paid at maturity, discounted using the corresponding simulated interest rate paths. The repayment amount is determined by the principal reduction factor, which is calculated using the simulated mortality rate paths over the bond's 5-year term.

\begin{figure}[th!]
    \centering
    \includegraphics[width=0.6\textwidth]{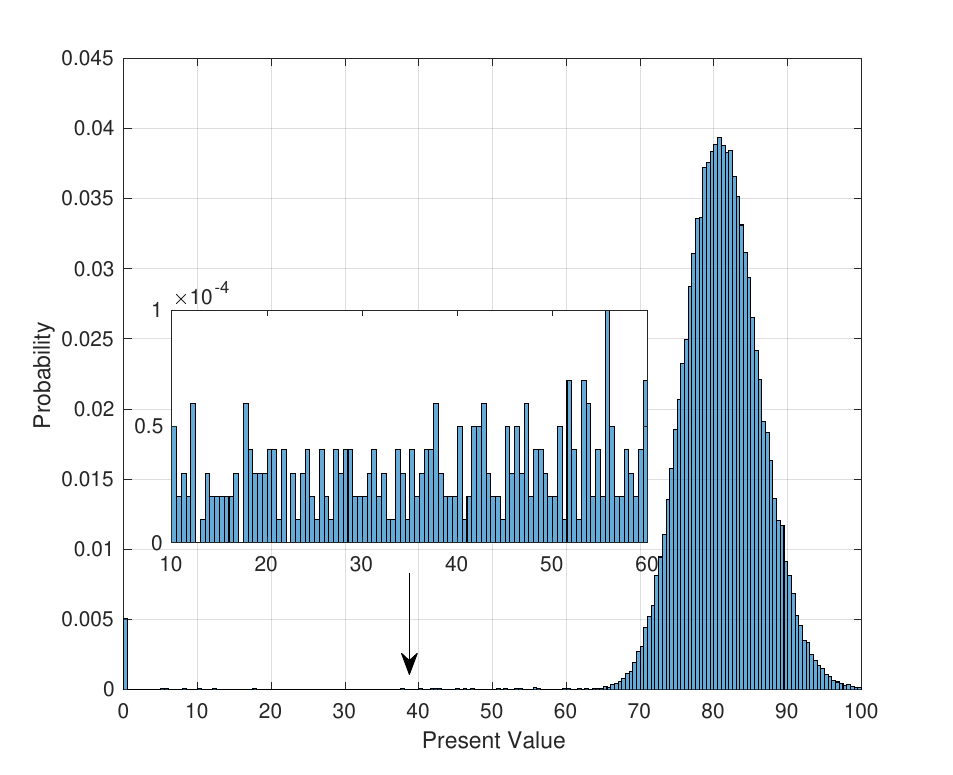}
    \caption{Histogram of simulated present values of principal repayment at time 0, computed under the physical measure $\mathbb{P}$ using simulated mortality and interest rate paths.}
    \label{fig:prf_histogram}
\end{figure}

The main histogram of Figure~\ref{fig:prf_histogram} shows the distribution of discounted principal repayments across all simulated mortality and interest rate scenarios. A large proportion of the mass is concentrated near full repayment, reflecting typical cases where the mortality index remains below the attachment point and investors recover the entire principal. To better examine tail risks, the figure includes a zoom-in on the lower repayment region, specifically values between 10 and 60. This region corresponds to extreme mortality scenarios, where high mortality rates trigger significant reductions in principal repayment. Lastly, a small but noticeable proportion of the mass is concentrated at zero repayment, indicating the likelihood of the most severe cases.

\begin{table}[htbp]
\centering
\begin{tabular}{ccccccc}
\toprule
Risk measures & Mean & Std & VaR$_{5\%}$ & CTE$_{5\%}$ & VaR$_{1\%}$ & CTE$_{1\%}$ \\
\midrule
PV of principal repayment & 80.59 & 8.27 & 72.74 & 60.80 & 66.40 & 19.95 \\
PV of total bond payouts  & 106.41 & 8.77 & 97.24 & 85.30 & 90.29 & 45.50 \\
\bottomrule
\end{tabular}
\caption{Risk measures for the present value (PV) of principal repayment and the present value of total bond payouts under the physical measure $\mathbb{P}$.}
\label{tab:PrincipalRepaymentRisk}
\end{table}

Table~\ref{tab:PrincipalRepaymentRisk} summarizes key numerical risk measures of the simulated principal repayment values at time 0. Specifically, we report the mean, standard deviation (Std), Value-at-Risk at the $5\%$ level (VaR$_{5\%}$) and $1\%$ level (VaR$_{1\%}$), and Conditional Tail Expectation at the $5\%$ level (CTE$_{5\%}$) and $1\%$ level (CTE$_{1\%}$). The mean repayment amount is approximately $80.59$, indicating that investors typically recover a substantial portion of the principal under common mortality scenarios. The tail risk measures offer a deeper insight into extreme scenarios. The VaR$_{5\%}$ value of $72.74$ indicates that there is a $5\%$ probability that principal repayment will fall below this threshold, while the CTE measures emphasize the severity of these tail scenarios. The CTE$_{5\%}$ at $60.80$ and the substantially lower CTE$_{1\%}$ at $19.95$ clearly show that extreme mortality events can lead to substantial losses.

\subsubsection{Total payouts}

Next, we analyze the bond's total payouts, including both coupon payments and principal repayments. All payouts are discounted to the present values at time 0 using the corresponding simulated interest rate paths. Principal repayments are adjusted according to the simulated mortality paths. Examining this result provides a comprehensive view into the overall financial performance of the bond under simulated scenarios.

\begin{figure}[th!]
    \centering
    \includegraphics[width=0.6\textwidth]{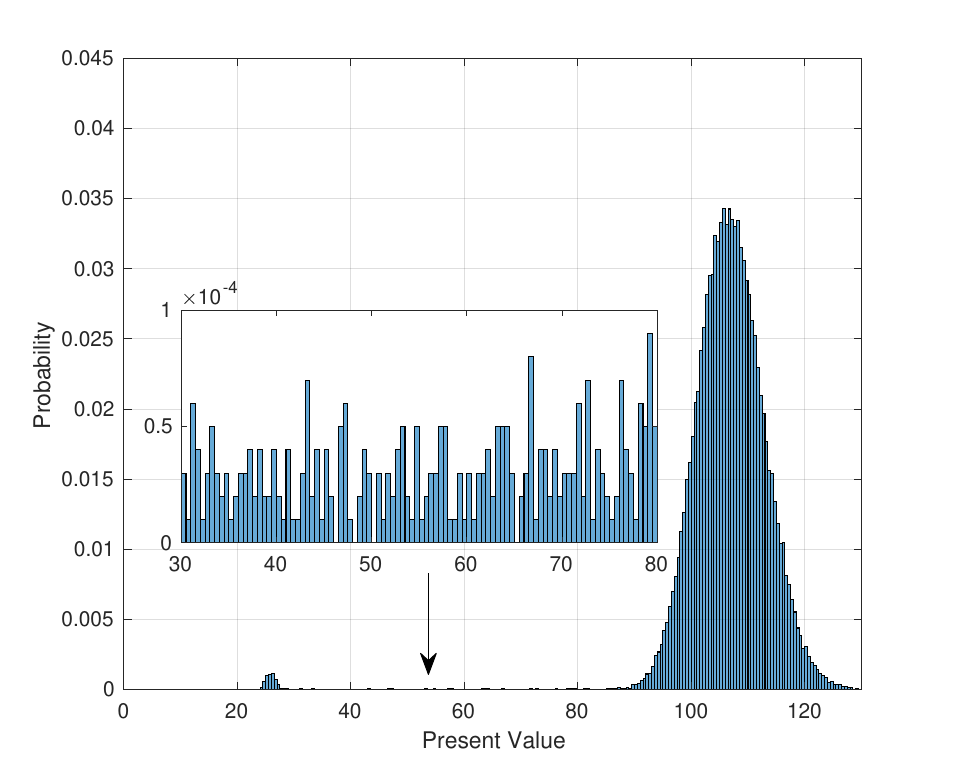}
    \caption{Histogram of simulated present values of total bond payouts (coupons plus principal) at time 0, computed under the physical measure $\mathbb{P}$ using simulated mortality and interest rate paths.}
    \label{fig:bond_payouts_histogram}
\end{figure}

Figure~\ref{fig:bond_payouts_histogram} displays the distribution of simulated bond payouts, reflecting both coupon payments and principal repayments discounted to time 0. The main histogram shows that most simulated scenarios result in total payouts above the face value ($F=100$), driven by the relatively high coupon rate. The zoom-in region highlights rare but severe outcomes, corresponding to significant principal reductions triggered by extreme mortality scenarios. Lastly, a small number of simulated scenarios correspond to complete principal loss, resulting in total payouts of only about 20–30 from coupon payments alone.

Table~\ref{tab:PrincipalRepaymentRisk} provides numerical risk measures summarizing the distribution of the bond's total payouts at time 0. The mean payout is approximately 106.41, indicating an expected outcome slightly above par value due to the high coupon accruals. The standard deviation of 8.77 reflects noticeable variability arising from mortality and interest rate uncertainties. The VaR$_{5\%}$ of 97.24 and VaR$_{1\%}$ of 90.29 highlight potential losses in adverse mortality scenarios, while the corresponding CTE values underscore the severity of these tail events, with CTE$_{5\%}$ at 85.30 and CTE$_{1\%}$ dropping to only 45.50.

\subsection{Sensitivity analysis}

In this subsection, we conduct a series of sensitivity analyses to better understand how the bond’s outcomes respond to changes in input assumptions. We examine key quantities, including loss metrics, payout measures, and the fair coupon rate. Three categories of parameters are considered: (i) bond design features, such as attachment and detachment points; (ii) model parameters estimated under the physical measure; and (iii) pricing parameters calibrated under the risk-neutral pricing measure. Each scenario builds on the baseline results presented earlier. Unless otherwise specified, all other inputs and bond settings remain fixed at their baseline values.

\subsubsection{Sensitivity to bond parameters}
We first analyze how bond structural parameters, namely, the attachment point, detachment point, and bond duration, influence the fair coupon rate of a catastrophe mortality bond.

\begin{figure}[th!]
\centering
\includegraphics[width=0.6\textwidth]{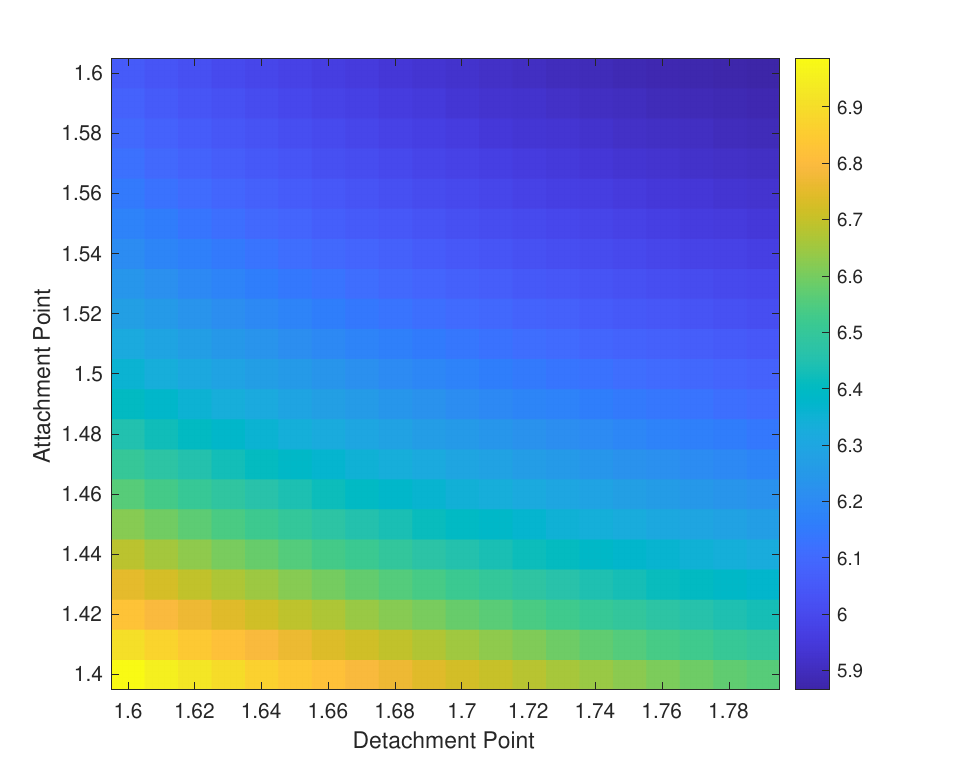}
\caption{Heatmap showing the fair coupon rate as a function of the attachment point and detachment point, calculated under the risk-neutral pricing measure $\mathbb{Q}$.}
\label{fig:sens_apdp}
\end{figure}

Figure~\ref{fig:sens_apdp} presents a heatmap of the fair coupon rate as a function of the attachment and detachment points, expressed as multiples of the mortality experience in the bond issue year 2024. As expected, the coupon rate decreases monotonically as either the attachment point or the detachment point increases. A higher attachment point reduces the likelihood of triggering principal reduction, while a higher detachment point limits the severity of potential principal losses.

More interestingly, the sensitivity shown in Figure~\ref{fig:sens_apdp} is not uniform across the parameter space. We observe that changes in the detachment point have a more pronounced effect when the attachment point is low. As the attachment point increases, the fair coupon rate becomes less sensitive to further increases in the detachment point. This pattern reflects a diminishing marginal effect of expanding the reduction coverage once the probability of triggering any principal reduction becomes sufficiently small.

\begin{table}[htbp]
\centering
\begin{tabular}{ccccc}
\toprule
Maturity Term & PFL & CEL & EL & Coupon Rate \\
\midrule
5 Years (Baseline) & 1.06\% & 70.75\% & 0.75\% & 5.85\% \\
4 Years & 0.78\% & 70.59\% & 0.55\% & 5.72\% \\
3 Years & 0.41\% & 67.51\% & 0.27\% & 5.51\% \\
2 Years & 0.15\% & 59.27\% & 0.09\% & 5.22\%  \\
1 Year  & 0.00\% & N/A & N/A & 4.85\% \\
\bottomrule
\end{tabular}
\caption{Loss metrics of the principle repayment and fair coupon rates for catastrophe mortality bonds with varying maturity terms, calculated under the risk-neutral measure $\mathbb{Q}$.}
\label{tab:coupon_rates}
\end{table}

Table~\ref{tab:coupon_rates} reports pricing metrics for catastrophe mortality bonds with maturity terms from 1 to 5 years, including the principal repayment's probability of first loss (PFL), conditional expected loss (CEL), expected loss (EL), and the fair coupon rate. All values are based on 10,000 simulations under the calibrated risk-neutral measure $\mathbb{Q}$. As the bond term shortens, both the probability and severity of losses decline, reflecting reduced exposure to extreme mortality events. This further leads to lower risk premiums, as reflected in steadily decreasing coupon rates.

The fair coupon drops from 5.85\% for a 5-year bond to 4.85\% for a 1-year bond. For the 1-year bond, no loss events were observed, so CEL and EL are not reported. The corresponding coupon rate thus reflects only interest rate risk. Although the CEL remains nearly unchanged between the 4-year and 5-year bonds, the expected loss decreases slightly, consistent with the reduction in PFL from 1.06\% to 0.78\%. These results showcase how bond duration affects loss likelihood, severity and coupon rate of catastrophe mortality bonds.

\subsubsection{Sensitivity to model parameters}

We now investigate the sensitivity of bond outcomes to various model parameters. In particular, we consider the following scenarios:
\begin{itemize}
    \item \textbf{Scenario 1:} The instantaneous correlation parameter is set to zero (i.e., $\rho = 0$), implying no correlation between the interest rate process $r_t$ and the excess mortality rate process $\mu_t$.
    \item \textbf{Scenario 2:} The Hurst parameters for both $r_t$ and $\mu_t$ are set to 0.5 (i.e., $H_1 = 0.5$ and $H_2 = 0.5$), indicating the absence of long-range dependence in either process.
    \item \textbf{Scenario 3:} The variance parameter of the interest rate process is doubled (i.e., $\sigma_1^2 = 2 \times 1.24565^2$), reflecting higher uncertainty in interest rate dynamics.
    \item \textbf{Scenario 4:} The variance parameter of the excess mortality rate process is doubled (i.e., $\sigma_2^2 = 2 \times 0.00286^2$), reflecting higher uncertainty in mortality dynamics.
    \item \textbf{Scenario 5:} The pricing parameter associated with $r_t$ is increased by 50\% (i.e., $\gamma_1 = 1.50 \times 3.8701$), representing a higher risk premium for interest rate risk.
    \item \textbf{Scenario 6:} The pricing parameter associated with $\mu_t$ is increased by 50\% (i.e., $\gamma_2 = 1.50 \times 1.0620$), representing a higher risk premium for mortality risk.
\end{itemize}

Table~\ref{tab:sensitivity_results} reports the fair coupon rate of the bond under each scenario, along with summary statistics and loss metrics for the present value of total bond payouts. For reference, the corresponding quantities under the baseline scenario are also included. As in the baseline case, the fair coupon rates are determined using simulations under the risk-neutral measure $\mathbb{Q}$, while the payout risk measures are computed under the physical measure $\mathbb{P}$.

\begin{table}[h!]
\centering
\caption{Sensitivity analysis results for the scenarios under consideration. Each row reports summary statistics of the present value of total bond payouts (mean, standard deviation, VaR, CTE), loss metrics of the principle repayment (PFL, CEL, EL), and the fair coupon rate.}
\label{tab:sensitivity_results}
\begin{tabular}{lcccccccc}
\toprule
Scenario & Mean & Std & VaR$_{5\%}$ & CTE$_{5\%}$ & PFL & CEL & EL & Coupon Rate \\
\midrule
Baseline   & 106.41 & 8.77 & 97.24 & 85.30 & 1.06\% & 70.75\% & 0.75\% & 5.85\% \\
Scenario 1 & 107.50 & 8.59 & 98.26 & 87.11 & 0.97\% & 70.97\% & 0.69\% & 6.08\% \\
Scenario 2 & 106.51 & 4.19 & 100.87 & 98.18 & 0.14\% & 65.38\% & 0.09\% & 5.77\% \\
Scenario 3 & 109.11 & 10.62 & 95.91 & 84.06 & 1.06\% & 70.75\% & 0.75\% & 6.42\% \\
Scenario 4 & 105.93 & 19.76 & 31.54 & 29.66 & 7.36\% & 83.84\% & 6.17\% & 6.73\% \\
Scenario 5 & 109.42 & 8.83 & 100.09 & 88.16 & 1.06\% & 70.75\% & 0.75\% & 6.53\% \\
Scenario 6 & 106.79 & 8.78 & 97.60 & 85.66 & 1.06\% & 70.75\% & 0.75\% & 5.93\% \\
\bottomrule
\end{tabular}
\end{table}

Scenarios 1 and 2 examine the effects of removing key dependence structures between mortality and interest rates. When correlation is removed in Scenario 1, the coupon rate increases noticeably and payout risk measures shift relative to the baseline, indicating that short-term co-movement between the two processes materially influences pricing. In contrast, removing long-range dependence in Scenario 2 results in a small reduction in the coupon rate and noticeably lower tail risk, suggesting that long-memory effects remain important for characterizing extreme outcomes. These findings confirm the need of incorporating both correlation and long-memory features into the joint modeling framework, as each plays a distinct role in shaping risk and return in catastrophe mortality bonds.

Scenarios 3 and 4 examine the effects of higher volatility in the interest rate process and the excess mortality rate process, respectively. Doubling the variance of interest rate in Scenario 3 leads to a noticeable increase in the coupon rate and greater variability in the bond payouts, indicating that interest rate volatility has a meaningful influence on bond pricing and return uncertainty. In contrast, doubling the variance of mortality rate in Scenario 4 results in a dramatic increase in tail risk, accompanied by a significantly higher coupon rate. This result suggests that mortality volatility is a comparatively stronger driver of both pricing and risk in catastrophe mortality bonds.

Scenarios 5 and 6 assess the impact of higher risk premiums for interest rate risk and mortality risk, respectively. In Scenario 5, increasing the interest rate risk premium leads to a noticeable increase in the coupon rate, reflecting the higher compensation required by investors for bearing additional interest rate risk. In Scenario 6, increasing the mortality risk premium by the same proportion produces a more moderate increase in the coupon rate. Overall, interest rate risk premiums have a stronger influence on investor compensation than mortality risk premiums. Lastly, the loss metrics under both scenarios remain unchanged relative to the baseline because these metrics are irrelevant to how risks are compensated under the risk-neutral measure $\mathbb{Q}$.

\section{Conclusion}

In conclusion, this paper presents a novel bivariate framework grounded in a mixed fractional Brownian motion to jointly capture the long-range dependence and instantaneous correlation between interest rate and excess mortality rate. By deriving the closed-form pricing formula for the zero-coupon bond and the pricing method for the mortality catastrophic bond under risk-neutral measure, we provide a tractable approach for valuing mortality-linked security in the environment characterized by persistent and correlated risks. 

The proposed sequential calibration method facilitates the practical estimation of the model parameters from real-world data. Our numerical experiments reveal that both interest rate and mortality risks are significant risk drivers for catastrophe mortality bonds. We demonstrate that the incorporation of long-term memory and cross-risk dependence should be considered for the design of MLS. These findings offer valuable insights for actuaries and risk managers who seek innovative ways to better design and price mortality-linked instruments and the related risk management strategies.

We acknowledge that there are limitations in our current work, and there remain research problems that are worth more in-depth explorations in the near future. First, excess mortality depends on the choice of baseline mortality. The risk associated with baseline mortality has not been explicitly considered in the design of MLS. It will be interesting to evaluate the impact of the choice of baseline mortality rates on the MLS design. Alternatively, we can refine our stochastic differential equations to model the mortality rate directly by adding features such as seasonality. 

Second, our work can be extended to study the mortality index and explore the design of other MLS based on mortality indexes. Lastly, while the current framework assumes a constant instantaneous correlation between mortality and interest rates, it would be valuable to explore more flexible correlation structures and develop corresponding calibration techniques, particularly in light of the empirical challenge posed by having only a single observed trajectory for each process.

\bibliographystyle{plain}
\bibliography{mfBm}


\end{document}